\documentclass{article}
\usepackage{iclr2027_conference,times}

\input{commands.sty}

\title{Locally Sparsified, Globally Near-Optimal: \\ Matching under Independent Vertex Arrivals}
\author{Sara Ahmadian \\
Google Research \\
\texttt{sahmadian@google.com}
\And
Edith Cohen \\
Google Research \\
Tel Aviv University \\
\texttt{edith@cohenwang.com}
\And
Mohammad Roghani \\
Google Research \\
\texttt{roghani@google.com}
}
\date{}
\hypersetup{
    pdftitle={Locally Sparsified, Globally Near-Optimal: Matching under Independent Vertex Arrivals},
    pdfauthor={Sara Ahmadian, Edith Cohen, Mohammad Roghani}
}

\iclrfinalcopy
\begin{document}
\maketitle
\lhead{Preprint}

\begin{abstract}
Resource allocation systems often restrict each request to a short
list of options before coordinating assignments globally.
We study this separation in stochastic bipartite matching under independent vertex arrivals. Each request draws a state from its own known distribution, determining its compatible resources, and independently retains a menu of at most $k$ edges. A maximum matching is then computed on the retained graph.

We show that bounded local menus universally suffice for near-optimal
matching.  For every $\varepsilon>0$, there is a menu size
$k_\varepsilon$ depending only on $\varepsilon$
that preserves at least a $(1-\varepsilon)$ fraction of the expected maximum-matching size of the full realized graph.
Earlier guarantees required additional assumptions on how matching mass is distributed across edges; our result resolves the unrestricted case.

Moreover, the menus are simple to generate from any benchmark matching rule, either by weighted sampling according to the benchmark's edge
marginals, or by applying the benchmark to sampled realizations and
retaining the resulting partners.

Our proof constructs a near-optimal certificate inside the sparsifier
by combining a \emph{locally computable} surrogate for the large-marginal
edges with a fractional completion from sampled light edges.  The
surrogate nearly preserves the benchmark's value and endpoint loads
while controlling dependencies, which makes the statistical
light-edge completion possible.
\end{abstract}

\section{Introduction}
\label{sec:introduction}

Resource allocation often separates local candidate selection from
global coordination.  In a batched robotaxi dispatch system, each
request can identify a short list of compatible vehicles before the
platform jointly assigns vehicles to requests.  Similarly, a cloud
scheduler can receive a small set of eligible servers for each job
before resolving competition for server capacity.  These examples
motivate a basic question: how much information must each request
retain for the final allocation to remain nearly optimal?

The local sparsification model of \citet{ACR26} captures this
separation.  Resources form one side of a bipartite graph, and
requests form the other.  Each request observes its compatible
resources and retains a \emph{menu} of at most $k$ incident edges.
Its choice may use statistical knowledge of the instance and private
randomness, but cannot depend on the realized requests elsewhere in
the system.  A central coordinator then computes a maximum matching
using the retained edges.  The parameter $k$ measures the amount of
information passed from each local decision to the global
optimization: $n$ requests communicate at most $nk$ resource
identifiers.

We study this model under \emph{independent vertex arrivals}.
Each request independently draws its compatible neighborhood
from its own known distribution. The model includes both i.i.d.\ request types and independent
presence or absence of vertices with fixed neighborhoods.

Statistical knowledge can help a request select options that will
remain useful when other requests compete for the same resources.
However, each menu must be chosen without observing that competition.
An option that is dispensable in one realization may be essential in
another.  The question is whether a bounded number of local choices
can accommodate this uncertainty.  \citet{ACR26} established near-optimal preservation under a spread
condition on the fractional plan guiding the menus, and demonstrated
strong empirical performance on ride-hailing data and synthetic
instances, including instances outside the regime covered by their
spread-based analysis.  Their work left open whether the spread
condition reflects a genuine barrier to near-optimal preservation or
only a limitation of the analysis.  We ask this question for the
more general independent-arrival model.

\begin{question}[Near-optimal matching from bounded local menus]
\label{q:local-sparsification}
For every $\varepsilon>0$, is there a menu bound $k_\varepsilon$,
depending only on $\varepsilon$, such that every instance with
independent vertex arrivals admits a local menu rule preserving
at least a $(1-\varepsilon)$ fraction of the expected
maximum-matching size of the full realized graph?
\end{question}

\paragraph{Our results.}
We answer \Cref{q:local-sparsification} affirmatively.

\begin{theorem}[Informal; see \Cref{thm:status-product}]
For every $\varepsilon>0$, there is a finite $k_\varepsilon$
such that every instance with independent vertex arrivals admits
a local menu rule of size at most $k_\varepsilon$ whose retained
graph $H$ satisfies
\[
    \mathbb E[\nu(H)]
    \ge
    (1-\varepsilon)\mathbb E[\nu(\mathcal G)],
\]
where $\mathcal G$ is the full realized graph and $\nu(G)$ is the maximum-matching size of a graph $G$.
The expectations are over the arrivals and the menu randomness.
\end{theorem}

Note that the number of retained choices per request needed to preserve the expected global matching value depends only on the desired accuracy.
In particular, it is independent of the number of requests, resources, possible states, or neighborhood sizes.

The menu construction is simple: weighted sampling according to
the conditional edge marginals of any randomized benchmark matching
rule $M^\star$ preserves $(1-\varepsilon)\mathbb E[|M^\star|]$
for $k\ge k_\varepsilon$. Alternatively, only sample access is
needed: for each request, fix its observed state, independently
simulate $k$ realizations of the other requests, run $M^\star$
on each, and retain its at most $k$ distinct assigned partners.
Taking $M^\star$ to compute a maximum matching gives our
near-optimality guarantee. These constructions extend the
Monte Carlo approach explored empirically for i.i.d.\ arrivals
by \citet{ACR26} to general independent arrivals.

Our proof gives the quantitative bound
\[
    k_\varepsilon
    \le
    \exp\exp\!\left(
        O\!\left(\varepsilon^{-1}\log\frac1\varepsilon\right)
    \right).
\]
We show that dependence on the menu size is unavoidable.
We establish in \Cref{sec:lower-bound} a lower bound that applies \emph{for every local menu rule}, exhibiting independent-presence instances on which the expected loss is $\Omega(1/k)$ of the full maximum-matching value. Consequently, any universal $(1-\varepsilon)$ guarantee requires $k_\varepsilon=\Omega(1/\varepsilon)$.
The gap between this lower bound and our upper bound is
large, and our result is foundational and not intended to predict  menu sizes needed in practice.
 Our experiments in
\Cref{sec:experiments} examine the dependence on $k$ of hub and expanded instance families. They indicate that good preservation can already occur for small menu sizes, consistent with the strong empirical performance observed by \citet{ACR26} beyond the regime covered by their theory.

\paragraph{Related work.} In classical online bipartite matching, each arriving request
must be matched or discarded immediately. The RANKING algorithm of \citet{KVV90} fixes a uniformly random
priority order on the offline vertices and matches each request
to its highest-priority available neighbor.
It achieves the optimal $1-1/e$ competitive ratio under
adversarial arrivals.
Random-order arrivals allow improved guarantees without
knowledge of a demand distribution
\citep{KMT11,MY11, peng2025revisiting}.
A complementary line of work assumes requests are drawn
independently from a known distribution
\citep{FMMM09,BK10,MOS12,JL14,HS21,HSY22}.
In all three models, however, the worst-case competitive ratio
is bounded away from one: this limitation persists even with
known i.i.d.\ arrivals \citep{MOS12,CGPV25}.

Local sparsification permits the final matching to be deferred:
the irrevocable decision is which options to retain.
Unlike an online matching algorithm, each request chooses its
menu without observing earlier requests or current resource
availability.
Despite this restriction, our result shows that a menu size
depending only on $\varepsilon$ suffices to preserve a
$(1-\varepsilon)$ fraction of the expected offline optimum
under independent arrivals with known distributions.

Streaming matching likewise permits a final computation after
processing the input, and studies what can be preserved in
limited memory \citep{GKK12,Kapralov13, kapralov2021space}.
A streaming algorithm can use its accumulated memory when
processing each arrival.
Here each menu depends only on the known arrival distributions,
the request's own state, and its private randomness.
The statistical assumption also differs from random-order
streams: we sample request states independently, rather than
randomly permuting a fixed collection of requests.
Our formulation isolates the power of distribution-informed
local selection followed by global optimization.

A different stochastic matching literature studies sparse
queries to a graph whose edges are independently realized
\citep{BDHPS15,AKL19,BFHR19,BDH20, behnezhad2020stochastic, ABGR25}.
In its nonadaptive formulation, the query subgraph is chosen
before the edge realizations are observed.
Our menus are chosen after observing each request's realized
neighborhood, with a separate budget for each request.
At the same time, a single request state can jointly determine
many edges, so independence of individual edges is unavailable.

\paragraph{Proof ideas.}
We prove the guarantee by constructing an auxiliary fractional
matching on the retained graph whose expected value nearly matches
that of the benchmark. This certificate is used only in the
analysis: the algorithm simply forms the menus and computes a
maximum matching on the retained graph. The large-marginal edges of the benchmark are retained in the menu.
The main challenge is that the large- and small-marginal parts
of the benchmark share endpoint capacity. Although the benchmark
has the right value and endpoint loads, its unmatched vertices
may have global dependencies that obstruct concentration for
the sampled light edges.

To control these dependencies, we build on the connection between
local computation algorithms and stochastic matching developed
by \citet{ABGR25}, using the framework of \citet{RTVX11,ARVX12}.
We replace the large-marginal benchmark matching with a locally
computable surrogate that nearly preserves its value and endpoint
loads, while in--out query bounds control dependencies among
unmatched vertices. We adapt this approach from independently
realized edges to independent vertex arrivals, where resampling
one request can change an entire neighborhood.
The resulting residual capacity allows a fractional completion
from sampled light edges to recover nearly all of the remaining
benchmark mass. Bipartite matching integrality then gives an
integral matching of at least the certificate's value.
\Cref{sec:main-theorem} gives the formal theorem and proof roadmap.

\section{Product-arrival model}
\label{sec:model}

\begin{definition}[Product-arrival instance]
\label{def:product-arrival}
A \emph{product-arrival instance} $\mathcal I$ consists of a finite set
$V$ of offline vertices, an integer $n$, and, for each online position
$i\in[n]$, a finite state space $\mathcal T_i$, a distribution $p_i$
with support $\mathcal T_i$, and a neighborhood map
\[
    \Gamma_i:\mathcal T_i\to 2^V.
\]

A \emph{realized graph} $\mathcal G\sim\mathcal I$ is sampled by drawing
independently $T_i\sim p_i$ for each $i\in[n]$ and setting
\[
    \mathcal G
    =
    \bigl([n],V,\{(i,v):v\in\Gamma_i(T_i)\}\bigr).
\]
We call a pair $(i,t)$ with $t\in\mathcal T_i$ a \emph{state}, and a triple $(i,t,v)$ with $v\in\Gamma_i(t)$ a \emph{state-edge}.
\end{definition}

\paragraph{Sparsification via menus.}

For a graph $G$, let $\nu(G)$ denote its maximum-matching size.

\begin{definition}[Menu rule]
\label{def:menu-rule}
A \emph{menu rule of size $k$} for an instance $\mathcal I$ assigns to every state $(i,t)$ a
distribution $\mathcal D_{it}$ over subsets
$S\subseteq\Gamma_i(t)$ with $|S|\le k$.

Given a realized graph $\mathcal G$, the rule samples a \emph{retained graph} $H$
by drawing, independently for each $i$,
\[
    S_i\sim\mathcal D_{i,T_i},
\]
and setting
\[
    H
    :=
    \bigl([n],V,\{(i,v):v\in S_i\}\bigr)
    \subseteq \mathcal G.
\]
The value of the menu rule is $\mathbb E[\nu(H)]$,
where the expectation is over both the product arrivals and the menu randomness.
\end{definition}

The distributions $\mathcal D_{it}$ may depend arbitrarily on the known instance $\mathcal I$; the realized menu draws  use independent private randomness.

\paragraph{Benchmark marginals and derived menus.}

Fix a randomized benchmark matching rule $M^\star$.  We represent its
randomness by an independent random seed $\omega^\star$ and write
\[
    M^\star=M^\star(T,\omega^\star),
\]
where for every state profile $T$, $M^\star(T,\omega^\star)$ is a
matching of the realized graph $\mathcal G(T)$.  The seed
$\omega^\star$ is independent of the arrival states and of all menu
and auxiliary randomness used below.

For every state-edge $(i,t,v)$, define the conditional and unconditional
benchmark marginals
\[
    z_{itv}
    :=
    \Pr[(i,v)\in M^\star\mid T_i=t],
    \qquad
    x_{itv}
    :=
    p_i(t)z_{itv}.
\]
Since $M^\star$ is always a matching,
\[
    \sum_v z_{itv}\le1
    \quad\text{for every state $(i,t)$},
    \qquad
    \sum_{i,t}x_{itv}\le1
    \quad\text{for every $v\in V$}.
\]
Write
\[
    Z:=\sum_{i,t,v}x_{itv}
      =\mathbb E[|M^\star|],
\]
and, for a set $F$ of state-edges,
\[
    x(F):=\sum_{(i,t,v)\in F}x_{itv}.
\]
If $M^\star$ selects a maximum matching of every realized graph, then
\[
    Z=\mathbb E[\nu(\mathcal G)].
\]

\begin{definition}[Benchmark-derived menu rule]
\label{def:derived}
A menu rule of size $k$ is \emph{derived from $M^\star$} if, for every
state $(i,t)$ and $S\sim\mathcal D_{it}$:
\begin{itemize}[topsep=2pt,itemsep=2pt]
\item[(D1)] \emph{Coverage:} for every $v\in\Gamma_i(t)$,
\begin{equation}
\label{eq:menu-property}
    \Pr[v\notin S]\le e^{-kz_{itv}}.
\end{equation}
\item[(D2)] \emph{Nonpositive covariances:} for distinct $v,w\in\Gamma_i(t)$, $\Cov\!\left(\mathbf1\{v\in S\},\mathbf1\{w\in S\}\right)
    \le 0$.
\end{itemize}
\end{definition}

We will use the consequence
\begin{equation}
\label{eq:inclusion-lower-bound}
    \Pr_{S\sim\mathcal D_{it}}[v\in S]
    \ge (1-e^{-1})\min\{1,kz_{itv}\}
    \ge \tfrac58\min\{1,kz_{itv}\}.
\end{equation}
More generally, (D1) could be replaced by $\Pr[v\notin S]\le e^{-c k z_{itv}}$ for any fixed
constant $c>0$; this only changes universal constants in our statements and proofs below.

Two standard constructions satisfy these conditions.  A fixed-size
VarOpt$_k$ sample of $(z_{itv})_v$
\citep{Chao82,Srinivasan01,CohenDKLT11} keeps every edge with
$z_{itv}\ge1/k$ and has miss probability at most
$(1-kz_{itv})_+\le e^{-kz_{itv}}$, with negatively orthant dependent inclusion
indicators.  Alternatively, draw $k$ independent auxiliary realizations
of the environment conditional on $T_i=t$, run $M^\star$ on each, and let
$S_i$ be the set of partners assigned to $i$.  Then $|S_i|\le k$, an edge
is missed with probability $(1-z_{itv})^k\le e^{-kz_{itv}}$, and the
resulting occupancy indicators are negatively associated.

\section{Universal sparsification theorem and proof roadmap}
\label{sec:main-theorem}

\begin{theorem}[Universal sparsification for product arrivals]
\label{thm:status-product}
There are universal constants $C,\varepsilon_0>0$ such that, for every
$0<\varepsilon\le\varepsilon_0$, if
\[
    k_\varepsilon
    :=
    \left\lceil
    \exp\!\left(
        \exp\!\left(
            C\varepsilon^{-1}\log\frac C\varepsilon
        \right)
    \right)
    \right\rceil,
\]
then for every product-arrival instance, every matching rule $M^\star$,
and every size-$k_\varepsilon$ menu rule derived from $M^\star$,
\[
    \mathbb E[\nu(H)]
    \ge
    (1-\varepsilon)\mathbb E[|M^\star|].
\]
\end{theorem}

Taking $M^\star$ to be a maximum-matching selector gives
\[
    \mathbb E[\nu(H)]
    \ge
    (1-\varepsilon)\mathbb E[\nu(\mathcal G)].
\]

The remainder of this section gives a high-level roadmap of the proof with pointers to formal statements and technical arguments in the appendices.

Fix the benchmark $M^\star$ and write $Z:=\mathbb E[|M^\star|]$. Fix a benchmark-derived menu rule of size $k_\varepsilon$, and let
$H$ be its retained graph.  We carry out the proof using an auxiliary
accuracy parameter $\eta$.  Specifically, we construct a random feasible
fractional matching $Y$ supported on $H$, which we call our \emph{certificate}, such that
\[
    \mathbb E[|Y|]\ge (1-O(\eta))Z.
\]
By integrality of the bipartite matching polytope,
$\nu(H)\ge |Y|$ in every realization.  At the end of the proof, we set
$\eta$ to a sufficiently small constant multiple of $\varepsilon$.

Observe that we cannot simply use $M^\star\cap H$, the
benchmark restricted to retained edges, as our certificate. This is because light edges
(those with small conditional marginals) may carry a substantial fraction
of $Z$ even though each is retained with only a small probability.
Consequently, restricting the particular benchmark matching to $H$ can
lose most of this mass.  Recovering it requires using retained light edges
other than those selected by that realization of the benchmark.

\subsection{Crucial, light, and intermediate bands.}
We partition the state-edges according to their conditional benchmark marginals using two thresholds
\[
    0<\tau^-<\tau\le \frac14.
\]
The \emph{crucial}, \emph{intermediate}, and \emph{light}
state-edges are:
\begin{equation}
\label{eq:taubands}
    C_\tau
    :=
    \{z_{itv}\ge\tau\},
    \qquad
    I_{\tau^-,\tau}
    :=
    \{\tau^-<z_{itv}<\tau\},
    \qquad
    L_{\tau^-}
    :=
    \{z_{itv}\le\tau^-\}.
\end{equation}
The threshold $\tau$ is set to be large enough so that each realized crucial edge is in $H$ with probability at least $1-\eta$. The light threshold $\tau^-$ is set to be substantially smaller than $\tau$ as to facilitate statistical guarantees on the sampled (retained) light edges; the precise requirement appears in \Cref{lem:light-completion}.

Our certificate combines two parts that separately recover the crucial and the light parts of the benchmark mass:
\begin{itemize}
\item \textbf{Crucial.}
We construct an auxiliary matching $M_C$ on the realized crucial graph, which we call the \emph{crucial witness} (\Cref{sec:witnesspreview}). It nearly preserves the crucial benchmark value $x(C_\tau)$ and its endpoint loads are (nearly) upper bounded by those of the crucial benchmark $M^\star\cap C_\tau$. Additionally, the witness also has controlled local dependencies (\Cref{sec:local_roadmap}).
Note that we construct $M_C$ independently of the menu randomness.  Since each crucial edge is retained with probability at least $1-\eta$, intersecting $M_C$ with $H$ loses at most an $\eta$ fraction of its expected value.
\item \textbf{Light.}
We then construct a fractional matching on retained light edges whose endpoints are free in $M_C$ (\Cref{sec:fittinglight,sec:realizinglight} ). Together with $M_C\cap H$, this forms our certificate.
\item \textbf{Intermediate.}
The intermediate band $I_{\tau^-,\tau}$ is left unaccounted
for by our certificate. In \Cref{sec:thresholdsroadmap} we show that the two thresholds can be chosen so that  $x(I_{\tau^-,\tau})=O(\eta Z)$.
\end{itemize}

\subsection{Local computation as a dependency-control tool.}\label{sec:local_roadmap}

To recover the light mass from the sampled menus, we use
inverse-probability weights that preserve the target masses in expectation.  Variance bounds then control the loss from reducing overloaded endpoint weights to obtain a feasible fractional matching. Additionally, we must also restrict to vertices left free by $M_C$ and compensate for their availability. For the weighted sums to concentrate, we need to control the variance by accounting for the dependencies among the \emph{free-status indicators} of the crucial witness $M_C$, which indicate whether each vertex is unmatched in $M_C$.

Note that we cannot simply take $M^\star\cap C_\tau$ as the crucial witness. Although it has the desired value and endpoint loads, its free-status indicators may inherit arbitrary global correlations from $M^\star$.
We therefore construct $M_C$ using \emph{local computation}, which lets us control dependencies among its free-status indicators while preserving the required value and endpoint-load guarantees. We describe the local computation tools used for this purpose. See \cref{sec:prelim-local} for details.

 A local computation algorithm (LCA) \citep{RTVX11,ARVX12} represents a global randomized object through consistent local queries.  For example, rather than explicitly constructing an entire matching, a query at a vertex determines its matched/free status by probing only a small connected portion of the input graph.  Queries use shared randomness so that all answers are consistent with one global matching. We use the following two primitives.

\paragraph{In--out bounds to variance bounds.}
On a fixed input graph, each query reads randomness from only a bounded set of locations.  Overlap between the locations read by different queries controls the covariances of their outputs. Bounding this overlap therefore yields variance bounds for the weighted sums of free-status indicators.  We measure locality in two directions.  The \emph{out} complexity $R^+$ bounds how many input vertices a single query can probe, while the \emph{in} complexity $R^-$ bounds, for any fixed input vertex, the expected number of query roots whose executions probe that vertex. Informally, out locality asks how much one output can depend on, whereas in locality asks how many outputs can depend on one input.  Together these bounds control aggregate dependence: an in--out profile $(R^+,R^-)$ bounds the total overlap among the information read by different queries, yielding covariance and variance bounds for sums of local outputs (\Cref{lem:prelim-fixed-correlation}).  This in--out viewpoint was developed for (a different model of) stochastic matching by \citet{ABGR25}.

\paragraph{Local Weighted Matching}  We construct the crucial matching using a near-optimal local weighted-matching rule (\Cref{thm:prelim-matching}) with two additional properties: (i) in--out bounds that are independent of graph
size (depending only on the maximum degree and the accuracy parameter), and (ii) edge-inclusion probabilities in the output matching are continuous in
the weights on each fixed graph.

\subsection{A crucial witness: from benchmark loads to a local matching.} \label{sec:witnesspreview}
For an arrival realization $T$, let
\[
    \mathcal G_C(T)
    :=
    \bigl([n],V,
    \{(i,v):(i,T_i,v)\in C_\tau\}\bigr)
\]
be the realized crucial graph.  We seek a matching on this graph that
nearly preserves the crucial benchmark value, does not substantially
exceed its endpoint loads, and has sufficiently localized dependencies. We refer to this matching as the \emph{crucial witness}.

The construction (\Cref{sec:crucial-witness}) proceeds in two stages.  First, we cap the realized crucial graph to obtain degree bounds, which are needed to bound the in--out complexity of the local matching rule.
Every online vertex has at most $1/\tau$ crucial neighbors, but an
offline vertex can have arbitrarily large realized crucial degree.
We delete all edges incident to an offline vertex whose crucial degree
exceeds
\[
    D=\Lambda_\eta(\tau)
      =O\!\left(\frac1\tau+\log\frac1\eta\right),
\]
and denote the resulting graph by $\mathcal G_C^D$.

Let $M_C^{\star,D}$ be the restriction of the crucial part of $M^\star$ to
$\mathcal G_C^D$, let
\[
    Q:=\mathbb E[|M_C^{\star,D}|],
\]
and let $q$ denote its endpoint marginals, indexed by online
state-events and offline vertices.  We show (\Cref{lem:capping}) that capping preserves almost all of the crucial benchmark value:
\[
    Q\ge (1-O(\eta))x(C_\tau)
    \qquad\text{and}\qquad
    \Delta(\mathcal G_C^D)\le D.
\]

Capping also makes the graph \emph{stable under resampling}: changing one
arrival state modifies only a controlled portion of
$\mathcal G_C^D$.  This property lets us extend the fixed-input
correlation bounds supplied by the in--out framework to the randomness
of the product arrivals.  Combining capping stability with the locality
of the matching construction yields a decorrelation bound
(\Cref{thm:prelim-transfer}), which controls the variance of weighted sums of
free-status indicators.

Note that capping alone does not remove the benchmark's global correlations:
$M_C^{\star,D}$ may still inherit arbitrary dependencies from $M^\star$.  We use
$M_C^{\star,D}$ only as a reference for its value and endpoint marginals.
In \Cref{subsec:witness-construction}, we replace
$M_C^{\star,D}$ by an auxiliary matching $M_C$ on
$\mathcal G_C^D$, constructed with fresh randomness and controlled
in--out locality.

The construction
uses weighted matching to discourage endpoint loads from exceeding
the capped-benchmark loads $q$.
For a candidate vector $y$ of endpoint marginals, we assign weights to the crucial state-edges as a function of $y$ and the target loads $q$, and apply our local weighted-matching primitive to the realized capped crucial graph.  Let $\mathcal R(y)$ denote the resulting vector of endpoint marginals.

We choose the edge weights so that the induced map $\mathcal R$ is
continuous, and Brouwer's theorem supplies a fixed point
$y^\star=\mathcal R(y^\star)$.
The crucial witness $M_C$ is the matching
produced at this fixed point.

The weights penalize overloaded endpoints increasingly strongly.
In particular, once the candidate load at an endpoint $u$ reaches
$(1+\eta^2)q_u$, its penalty saturates and no incident edge is selected
by the response.  Such an overload therefore cannot occur at a fixed
point.  Together with near-optimality of the weighted-matching
response, this yields a witness that nearly preserves the capped
crucial value while controlling every endpoint load.  Writing
$y:=y^\star$ for its endpoint marginals,
\begin{equation}
\label{eq:overfitp}
    \mathbb E[|M_C|]\ge(1-\eta)Q,
    \qquad
    y_u\le(1+\eta^2)q_u
    \quad\text{for every endpoint }u.
\end{equation}
Its free-status indicators also have the controlled in--out locality
needed for the light-completion argument (\Cref{sec:realizinglight} ).

Note that $M_C$ is  constructed
without looking at the sampled menus and need not itself be contained
in $H$.  Only afterward do we intersect it with $H$.  Every edge of
$M_C$ is crucial, so, conditional on the arrival states and on $M_C$,
the coverage property (D1) gives
$\Pr[e\notin H]\le e^{-k\tau}$
for every $e\in M_C$.
Hence, once $k\tau\ge\log(1/\eta)$,
$
    \mathbb E[|M_C\cap H|]
    \ge
    (1-\eta)\mathbb E[|M_C|].
$
Combining this with the witness value guarantee and
$Q\ge(1-\eta)x(C_\tau)$  gives
$
    \mathbb E[|M_C\cap H|]
    \ge
    x(C_\tau)-O(\eta Z).
$

\subsection{Fitting the light benchmark to the witness.} \label{sec:fittinglight}
We next prepare the light benchmark mass for completion on endpoints
left free by $M_C$.  This step, formalized in
\Cref{subsec:light-fit}, uses only the marginal probabilities that
endpoints are free in $M_C$; it does not depend on which endpoints are
free in a particular realization.

The initial light \emph{plan} is the restriction of the benchmark
marginal vector $x$ to $L_{\tau^-}$.  Since the light part and
$M_C^{\star,D}$ are disjoint subsets of the benchmark matching, this
plan is feasible in the residual marginal capacities left by
$M_C^{\star,D}$: at most $p_i(t)-q_{it}$ at each online state
$(i,t)$ and $1-q_v$ at each offline vertex $v$.  The witness $M_C$
instead leaves residual capacities $p_i(t)-y_{it}$ and $1-y_v$.

The pointwise load guarantee \eqref{eq:overfitp} implies that the benchmark residual capacity can exceed the residual capacity left by $M_C$ by at most an $\eta^2$ fraction of the full endpoint capacity. Consequently, at every endpoint whose residual availability under $M_C$ is at least $\eta$, the benchmark residual capacity is at most a $1+\eta^2/\eta=1+\eta$ factor larger than the witness residual capacity.

We therefore discard light edges incident to endpoints whose residual availability is below $\eta$.  Such endpoints carry only $O(\eta Z)$ benchmark light mass in total.  On the remaining endpoints, dividing the light plan by $1+\eta$ makes it feasible in the residual capacities left by $M_C$.

The resulting subplan $x^0$ is supported on $L_{\tau^-}$, satisfies
$0\le x^0_{itv}\le x_{itv}$, has
\[
    |x^0|\ge x(L_{\tau^-})-O(\eta Z), \qquad
    \sum_v x^0_{itv}\le p_i(t)-y_{it},
    \qquad
    \sum_{i,t}x^0_{itv}\le1-y_v.
\]
Moreover, every endpoint in the support of $x^0$ has residual
availability at least $\eta$.  Write
$
    z^0_{itv}:=\frac{x^0_{itv}}{p_i(t)}
$
for its corresponding conditional masses. In the next step we show how to recover nearly all of $x^0$ from the retained light edges.

\subsection{Realizing the light residual.} \label{sec:realizinglight}
We now use the light edges retained by the original menu rule to realize the fitted analysis plan $x^0$ on endpoints left free by $M_C$. See \Cref{subsec:light-sampling} for full details. This requires compensating for
both menu sampling and endpoint availability.

First consider menu sampling.  Conditional on $T_i=t$, assign the
Horvitz--Thompson weight \citep{HT52}
\[
    W_{iv}
    :=
    z^0_{itv}
    \frac{\mathbf1\{v\in S_i\}}
         {\Pr[v\in S_i\mid T_i=t]}.
\]
Thus
\[
    \mathbb E[W_{iv}\mid T_i=t]=z^0_{itv}.
\]
These weights recover the fitted light masses in expectation.
When $1/k\le\tau^-$, the coverage property (D1) also bounds their
conditional second moments by $O(\tau^-)z^0_{itv}$, while (D2)
gives nonpositive covariances between distinct weights from the
same menu.

Next we restrict to endpoints free in $M_C$.  Let $F_i$ and $F_v$
be the indicators that online vertex $i$ and offline vertex $v$
are free, and write
\[
    a_{it}:=\Pr[F_i=1\mid T_i=t]
           =1-\frac{y_{it}}{p_i(t)},
    \qquad
    a_v:=\Pr[F_v=1]=1-y_v.
\]
Simply multiplying $W_{iv}$ by $F_iF_v$ would lose light mass:
even if the free statuses were independent, this would introduce
an additional factor $a_{it}a_v$ in the expected weight.
We therefore start with the availability-normalized weights
\[
    U_{iv}
    :=
    W_{iv}\frac{F_iF_v}{a_{iT_i}a_v}.
\]
Each availability in these denominators is at least $\eta$ on
the support of $x^0$, by the residual-fitting step.

The normalization compensates for endpoint availability, but it
does not make these weights exactly unbiased.  Conditional on
$T_i=t$, the indicators $F_i$ and $F_v$ may be correlated, and
the free probability of $v$ may differ from its unconditional
value $a_v$.  Moreover, feasibility requires control of the total
weight incident to each free endpoint, not just the expected
weight of each edge.

This is where the local construction of $M_C$ is used.
The product-arrival decorrelation bound of
\Cref{thm:prelim-transfer} controls variances of weighted sums
of free-status indicators and the aggregate bias caused by
conditioning on an arrival state.  Taking $\tau^-$ sufficiently
small relative to the locality bounds at threshold $\tau$ makes
the resulting losses $O(\eta Z)$.

The light-completion argument in \Cref{lem:light-completion} divides these weights by $1+\eta$ to create slack, then scales down overloaded rows and columns in each realization.  The variance and conditioning-bias bounds show that the expected mass lost is only $O(\eta Z)$.  To control the initial value error, we also use the pointwise witness bound, which bounds the total offline matched mass $\sum_v y_v$ by $(1+\eta^2)Q$.
The result is a feasible fractional matching $\widehat Y$
supported on retained light edges with endpoints free in $M_C$,
such that
\[
    \mathbb E[|\widehat Y|]
    \ge |x^0|-O(\eta Z)
    \ge x(L_{\tau^-})-O(\eta Z).
\]
Since $\widehat Y$ uses only vertices free in $M_C$, it is compatible
with the surviving crucial matching $M_C\cap H$.  Their union is our
fractional matching certificate $Y$.  Bipartite matching
integrality therefore gives
\begin{equation}
\label{eq:roadmap-certificate}
    \mathbb E[\nu(H)]
    \ge
    x(C_\tau)+x(L_{\tau^-})-O(\eta Z)
    =
    Z-x(I_{\tau^-,\tau})-O(\eta Z).
\end{equation}

\subsection{Choosing the band thresholds.} \label{sec:thresholdsroadmap}
It remains to choose the thresholds $\tau$ and $\tau^-$. See \Cref{sec:thresholds} for full details. For $m=\Theta(1/\eta)$, we construct $\rho_0>\rho_1>\cdots>\rho_m$,
so that every consecutive pair $(\tau,\tau^-) =
    (\rho_j,\rho_{j+1})$
has enough scale separation for the light-completion argument.  The
corresponding intermediate bands are disjoint.  Since their total
benchmark mass is at most $Z$, one of the $m$ bands has mass
$O(\eta Z)$ by pigeonhole.  We use that pair in
\eqref{eq:roadmap-certificate}.

The scale separation required at each step makes the thresholds shrink
by a large power.  Iterating this recurrence for
$\Theta(1/\eta)$ levels gives
$
    \rho_m^{-1}
    \le
    \exp\exp\!\left(
        O\!\left(
            \eta^{-1}\log\frac1\eta
        \right)
    \right).
$
Taking $k\ge\rho_m^{-1}$ also guarantees the menu-coverage condition
needed for the crucial witness and the sampling scale needed for the
light completion.  Finally setting $\eta=\Theta(\varepsilon)$ proves
\Cref{thm:status-product}.

\section{Experiments: loss at moderate menu sizes}
\label{sec:experiments}

We examine the loss $1-\mathbb E[\nu(H)]/\mathbb E[\nu(\mathcal G)]$
at moderate menu sizes.  Our main experiments use VarOpt$_k$ menus
constructed from estimated conditional marginals of a maximum-matching
benchmark with random tie-breaking.  The estimates and evaluation use
independent arrival profiles.

\paragraph{Rates and constants.}
\Cref{fig:rate} considers growing hub families, in which hubs compete
for resources left free by dedicated arrivals, and a variant with
contested private resources.  We vary the number of hubs and the
presence probabilities; instance sizes grow with $k$.
Each point reports the largest estimated loss over the tested scale
grid.  Across these families, the loss is below $1.3\%$ at $k=32$
and $0.7\%$ at $k=64$, with $k$ times the loss between $0.08$ and
$0.42$ over the tested range.  These structured families already admit
an $O(1/k)$ guarantee (\Cref{prop:hub-upper,app:experiments}); the plots assess constants for the implemented sampling rule.

\paragraph{Sampling choices.}
The comparisons in \Cref{tab:rules} illustrate the value of marginal
information and of spreading choices across resources.  At $k=4$,
the reported loss on random heterogeneous arrivals is $8.1\%$ for
uniform menus versus $0.29\%$ for VarOpt; on the four-hub family,
top-$k$ menus lose $11.4\%$ versus $6.6\%$ for VarOpt.
Additional tests cover nested neighborhoods and random presence graphs,
and show sensitivity to the benchmark selector.
\Cref{app:experiments} gives the families, scale grids, estimation
procedure, and further comparisons.

\begin{figure}[t]
\centering
\includegraphics[width=0.62\textwidth]{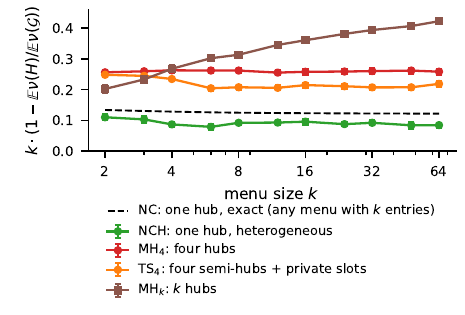}
\caption{Loss multiplied by $k$ for VarOpt$_k$ menus, maximized over
the tested scale grid.  Bars show twice the estimated evaluation
standard error.  Dashed: exact loss for the symmetric one-hub family.
Family definitions and grids are in \Cref{app:experiments}.}
\label{fig:rate}
\end{figure}

\section{Conclusion}
\label{sec:conclusion}

Bounded local menus preserve nearly all of the expected maximum-matching
value under independent vertex arrivals, with menu size depending only
on the desired accuracy.

The main open problem is to improve this dependence.  We conjecture
that the optimal worst-case loss is $\Theta(1/k)$, so menus of size
$O(1/\varepsilon)$ suffice for a $(1-\varepsilon)$ guarantee.
The hub class already admits this rate (\Cref{prop:hub-upper}).

A second direction is to remove knowledge of the arrival distributions,
for example under unknown i.i.d.\ arrivals or random-order vertex streams
\citep{KMT11,MY11}.  Can a single pass learn enough from an initial
sample to retain small menus for a near-optimal final matching?
The challenge is to learn useful statistics within the memory budget
and, for random-order streams, account for sampling without replacement.

\newpage
\section*{AI use statement}

We used generative AI tools, primarily ChatGPT (OpenAI),
Claude (Anthropic), and Gemini (Google), to explore proof strategies,
formulate and critically examine intermediate mathematical claims,
identify potential gaps, and assist with drafting and revising proofs.
We also used these tools for literature searches and improvements to implement and run experiments,
the manuscript's organization, notation, and exposition.
The authors take responsibility for the mathematical claims,
proofs, references, and final content of this work.

\section*{Reproducibility statement}
The model and sampling assumptions are specified in
\Cref{sec:model}, and the proof of the main theorem is developed
in \Cref{sec:prelim-local,sec:crucial-witness,subsec:light-fit,subsec:light-sampling,sec:thresholds,app:lca,app:decorrelation}.
The lower and upper bounds for hub instances are proved in
\Cref{sec:lower-bound}.
\Cref{app:experiments} describes the experimental families,
sampling rules, marginal estimation, evaluation procedure,
and uncertainty estimates.

\clearpage
\appendix
\section{Local computation primitives}
\label{sec:prelim-local}

The LCA model was introduced by
\citet{RTVX11} and further developed by \citet{ARVX12}; our analysis uses the in--out framework of \citet{ABGR25}.

\begin{definition}[Probe access and shared randomness]
\label{def:prelim-access}
An execution rooted at a vertex $u$ starts by probing $u$
and learns the input graph by probing vertices. Probing a vertex reveals its neighbors and local input data, and subsequent probes may only follow already revealed adjacencies; in particular the probed set is connected to $u$.  We call such an execution \emph{natural}.

Shared randomness $\pi$ is represented by independent \emph{dictionary marks}.  Each
mark is assigned to a vertex in its local support, called its
\emph{owner}, and may be read only after that owner has been probed.
Thus all randomness used by an execution is revealed locally along its
probe exploration.
\end{definition}

\begin{definition}[Local rule and query profile]
\label{def:prelim-rule}
A \emph{local rule} answers a query at every root $u$ by a natural
execution, returning an output $X_u=X_u(G,\pi)$. The answers $(X_u)_u$ jointly represent the global object computed by the rule.
Let $\mathrm{Read}(u)$ be the set of vertices probed by the execution rooted
at $u$.  Its \emph{query profile} is $(R^+,R^-)$ if
\[
    |\mathrm{Read}(u)|\le R^+
\]
pointwise for every root and every realization of the shared
randomness, while for every vertex $w$,
\[
    \mathbb E\!\left[
        \bigl|\{u:w\in\mathrm{Read}(u)\}\bigr|
    \right]
    \le R^-.
\]
We refer to $R^+$ and $R^-$ as the out-query and expected in-query
complexities, respectively.
\end{definition}

We use two LCA tools, stated next, with proofs deferred to
\Cref{app:lca}.

The first tool (\Cref{lem:prelim-fixed-correlation}) is a fixed-input correlation bound for local outputs.
For a fixed input graph and each vertex $u$, the expected number of
other queries whose read sets intersect $\mathrm{Read}(u)$ is at most
$R^+R^-$, where the expectation is over the shared randomness $\pi$.
These intersection probabilities control the absolute covariances
between the indicators $X_u$, yielding a bound on the variance of weighted sums of the form $\sum_u a_uX_u$. We apply this bound to the free-status indicators of the crucial matching. After accounting for the arrival randomness, this gives the variance bounds needed for the light-completion analysis in \Cref{lem:light-completion}.

\begin{restatable}[Fixed-input in--out correlation
{\citep[Lemmas~3.7--3.8]{ABGR25}}]{lemma}{fixedinputcorrelation}
\label{lem:prelim-fixed-correlation}
Fix a graph $G$ and a natural local rule whose randomness is grouped
into independent private tapes at their owner vertices.  Suppose its
query profile is $(R^+,R^-)$.  For roots $u,u'$, let
\[
    \delta(u,u')
    :=
    \Pr[
        \mathrm{Read}(u)\cap\mathrm{Read}(u')\ne\varnothing
    ].
\]
Then
\begin{equation}
\label{eq:prelim-delta-sum}
    \sum_{u'\ne u}\delta(u,u')
    \le R^+R^-
\end{equation}
for every $u$.  Moreover, if the query outputs are binary $X_u\in\{0,1\}$, then
\begin{equation}
\label{eq:prelim-cov-delta}
    |\Cov(X_u,X_{u'})|
    \le C\,\delta(u,u')
\end{equation}
for a universal constant $C$.  Consequently, for deterministic
coefficients $(a_u)$,
\begin{equation}
\label{eq:prelim-fixed-variance}
    \Var\!\left(\sum_u a_uX_u\right)
    \le
    C R^+R^-\sum_u a_u^2.
\end{equation}
\end{restatable}

The second tool, \Cref{thm:prelim-matching}, is a near-optimal
local weighted-matching rule. Local approximation schemes for maximum-weight matching were developed by \citet{EMR18}, building on the local-improvement framework of \citet{NO08}.
Our rule has bounds on both out-query and expected in-query complexity that are independent of graph size, and its output-edge marginals are continuous in the edge weights on each fixed graph. Continuity is needed for the fixed-point argument in
\Cref{sec:crucial-witness}.
The proof is deferred to \Cref{app:lca} and combines bounded weighted augmentations
\citep{PS04} with truncated random-greedy MIS \citep{ABGR25}.

We use auxiliary edge weights in our construction of the crucial witness, as a way to
penalize overused endpoints.  For edge weights
$w=(w_e)_{e\in E(G)}$, write
\[
    w(F):=\sum_{e\in F}w_e
    \qquad\text{and}\qquad
    \operatorname{OPT}_w(G)
    :=\max\{w(M):M\text{ is a matching of }G\}.
\]
When all edge weights equal one, these quantities reduce to
$|F|$ and $\nu(G)$, respectively.

\begin{restatable}[Local weighted matching rule]{theorem}
{weightedmatchingrule}
\label{thm:prelim-matching}
There are universal constants
$c_{\mathrm{lca}},C_{\mathrm{lca}}>0$ such that the following holds.
Fix $\eta\in(0,1/4)$ and $\gamma\in(0,1]$ with
$\gamma\ge c_{\mathrm{lca}}\eta$.
There is a natural local rule $\Psi$ with the following properties
for every finite graph $G$ of maximum degree at most $\Delta\ge2$
and every edge-weight vector $w\in[\gamma,1]^{E(G)}$.
The rule computes a matching
\[
    M=\Psi(G,\pi;w),
\]
where $\pi$ comprises all internal randomness, represented by
independent dictionary marks with local owners as in
\Cref{def:prelim-access}.

\smallskip
\noindent
\textnormal{(i) Approximation and locality.}
The output satisfies
\[
    \mathbb E_\pi[w(M)]
    \ge (1-\eta)\operatorname{OPT}_w(G).
\]
Its query profile satisfies
\[
    R^+,R^-\le\Delta^{a_\eta},
    \qquad
    a_\eta
    :=C_{\mathrm{lca}}\eta^{-2}\log^2\frac{e}{\eta}.
\]

\smallskip
\noindent
\textnormal{(ii) Continuous edge marginals.}
For fixed $G,\eta,\gamma$ and every $e\in E(G)$, the map
\[
    w\longmapsto
    \Pr_\pi[e\in\Psi(G,\pi;w)]
\]
is continuous on $[\gamma,1]^{E(G)}$.
The expectation and probability are over $\pi$, with the graph
and input weights fixed.
\end{restatable}

\section{The crucial witness}
\label{sec:crucial-witness}

\subsection{Capping the crucial graph}
\label{subsec:capping}

Fix $\tau\in(0,\tfrac14]$ and consider the crucial state-edges
$C_\tau$ defined in \Cref{eq:taubands}.  We first cap their realized degree.
Define the projected benchmark endpoint marginals
\[
    c_{it}
    :=
    \sum_{v:(i,t,v)\in C_\tau}z_{itv},
    \qquad
    c_v
    :=
    \sum_{i,t:(i,t,v)\in C_\tau}x_{itv}.
\]
Then
\[
    \sum_{i,t}p_i(t)c_{it}
    =
    \sum_vc_v
    =
    x(C_\tau).
\]

Every state $(i,t)$ has at most $\tau^{-1}$ crucial neighbors, since
$\sum_vz_{itv}\le1$ and every crucial edge has
$z_{itv}\ge\tau$.  Write
\[
    d:=\lceil\tau^{-1}\rceil .
\]
For an offline vertex $v$, let
\[
    N_v
    :=
    \sum_i
    \mathbf1\{(i,T_i,v)\in C_\tau\}
\]
be its realized crucial degree.  The summands are independent Bernoulli
variables, and
\begin{equation}
\label{eq:cap-expected-degree}
    \mathbb E[N_v]
    \le
    \frac{c_v}{\tau}
    \le
    \frac1\tau.
\end{equation}
Unlike the online degree, $N_v$ has no upper bound by a function of $\tau$. We show that we can cap the degree while
losing only an $O(\eta)$ fraction of the benchmark crucial value.

\begin{lemma}[Capping]
\label{lem:capping}
Fix $\eta\in(0,\tfrac12)$ and set
\[
    D=\Lambda_\eta(\tau)
    :=
    \Bigl\lceil\frac{2e}{\tau}\Bigr\rceil
    +\Bigl\lceil\log_2\frac1\eta\Bigr\rceil .
\]
Let $\mathcal G_C^D$ be the \emph{capped} crucial graph, that is,  $\mathcal G_C$ after deleting all
edges incident to an offline vertex with $N_v>D$, and define
\[
    M_C^{\star,D}:=(M^\star\cap C_\tau)\cap E(\mathcal G_C^D),
    \qquad
    Q:=\mathbb E[|M_C^{\star,D}|],
\]
with endpoint marginals
\[
    q_{it}:=\Pr[T_i=t,\ i\in V(M_C^{\star,D})],
    \qquad
    q_v:=\Pr[v\in V(M_C^{\star,D})].
\]
Then $\mathcal G_C^D$ is a deterministic function of $T$, and $M_C^{\star,D}$ is
always a matching of $\mathcal G_C^D$.  Moreover:
\begin{itemize}[topsep=2pt,itemsep=2pt]
\item[(i)] in every realization, the maximum degree is bounded by $D$
\begin{equation}
\label{eq:new-cap-degree}
    \Delta(\mathcal G_C^D)\le D=\Lambda_\eta(\tau);
\end{equation}

\item[(ii)] if $T,T'$ differ only in coordinate $i$, then every changed
capped edge is incident to one of the at most $2d$ offline vertices
adjacent to $i$ in one of the two realized crucial graphs, and
\begin{equation}
\label{eq:cap-stability}
    |E(\mathcal G_C^D(T))\triangle E(\mathcal G_C^D(T'))|
    \le 2dD.
\end{equation}
If $X_i(T,T')$ is the set of vertices whose capped neighborhood changes,
then
\begin{equation}
\label{eq:cap-vertex-stability}
    |X_i(T,T')\cup\{i\}|\le1+2d(D+1),
\end{equation}
and $X_i(T,T')$ lies within distance two of $i$ in the union of the two
realized crucial graphs;

\item[(iii)] the retained benchmark value satisfies
\begin{equation}
\label{eq:new-cap-value}
    Q\ge(1-\eta)x(C_\tau);
\end{equation}

\item[(iv)] the endpoint marginals satisfy
\[
    \sum_{i,t}q_{it}=\sum_vq_v=Q,
    \qquad
    q_{it}\le p_i(t)c_{it},
    \qquad
    q_v\le c_v.
\]
\end{itemize}
\end{lemma}

\begin{proof}
The graph $\mathcal G_C^D$ depends only on the realized states.  An
offline vertex either retains its $N_v\le D$ crucial edges or is
isolated, while every arrival has at most $d\le D$ crucial neighbors;
this proves \eqref{eq:new-cap-degree}.  Since $M_C^{\star,D}$ is obtained by
intersecting $M^\star$ with this graph, it is always a matching.

For the value bound, let $A_v$ be the event that $M^\star$ contains
a crucial edge incident to $v$.  Since $M^\star$ is a matching,
\[
    x(C_\tau)=\sum_v\Pr[A_v].
\]
The capping rule retains all crucial edges at $v$ when $N_v\le D$ and
isolates $v$ when $N_v>D$.  Thus
\[
    x(C_\tau)-Q
    =
    \sum_v \Pr[A_v\cap\{N_v>D\}]
    \le
    \sum_v \Pr[N_v>D].
\]

For every integer $r\ge1$, independence of the Bernoulli summands in
$N_v$ and \eqref{eq:cap-expected-degree} give
\[
    \Pr[N_v\ge r]
    \le\frac{\mathbb E[(N_v)_r]}{r!}
    \le\frac{c_v}{r!\tau^r},
\]
where we used $c_v\le1$.  Taking $r=D+1$ and summing over $v$ yields
\[
    x(C_\tau)-Q
    \le
    \frac{x(C_\tau)}{(D+1)!\tau^{D+1}}
    \le
    \eta x(C_\tau),
\]
because $D+1\ge2e/\tau$ and $D+1\ge\log_2(1/\eta)$.  This proves
\eqref{eq:new-cap-value}.

For stability, changing one state modifies only edges incident to the
arrival $i$, so only the at most $2d$ corresponding offline neighbors
can change their count.  At such a vertex, if the cap status does not
change then at most the edge to $i$ changes; if it flips, the two counts
are $D$ and $D+1$, so exactly the $D$ capped edges on the uncrossed side
can change.  This proves \eqref{eq:cap-stability}.  The endpoints of
these changed edges are $i$, the affected offline vertices, and at most
$D$ additional arrivals per affected offline vertex, giving
\eqref{eq:cap-vertex-stability} and the radius-two claim.

Finally, the two sums of the $q$'s count the two endpoints of the edges
of $M_C^{\star,D}$.  Conditional on $T_i=t$, the probability that $M^\star$ uses a
crucial edge at $i$ is $c_{it}$, so $q_{it}\le p_i(t)c_{it}$; similarly
$q_v\le c_v$.
\end{proof}

From now on set $\eta:=c_0\varepsilon$, where the universal constant
$c_0>0$ is chosen sufficiently small at the end, after fixing the
constants in the witness and light-completion bounds.  We choose
$\varepsilon_0$ sufficiently small that $\eta$ satisfies the universal
upper bounds required below.

\subsection{Constructing the local crucial witness}
\label{subsec:witness-construction}

We store each arrival state $T_i$ at vertex $i$.
The scores of its incident state-edges are evaluated there from
$T_i$ and the deterministic instance parameters.
Auxiliary marks have fixed local owners, as in
\Cref{def:prelim-access}, and are sampled independently of the arrival states and the private randomness used to draw menus.  This convention will also be used when coupling executions under a change of one arrival state in \Cref{app:decorrelation}.

\begin{lemma}[Crucial witness]
\label{lem:witness}
There is a universal constant $C_{\rm loc}>0$ such that the following
holds for every $0<\eta\le1/4$ and $\tau\in(0,\tfrac14]$.
Let $q_{it}$, $q_v$, and $Q$ be the capped-benchmark marginals and
value of \Cref{lem:capping}, with capping parameter $\eta$.
There is a random matching $M_C$ of the realized capped crucial graph
$\mathcal G_C^D$, a function of the arrival states and of a shared
dictionary $\pi$ independent of both the arrival states and the menu
randomness, whose endpoint marginals
\[
    y_{it}:=\Pr[T_i=t,\ i\in V(M_C)],
    \qquad
    y_v:=\Pr[v\in V(M_C)]
\]
satisfy:
\begin{itemize}[topsep=2pt,itemsep=2pt]
\item[(a)] \emph{value:}
\[
    \mathbb E[|M_C|]\ge(1-\eta)Q;
\]
\item[(b)] \emph{endpoint-load control:}
\[
    y_{it}\le(1+\eta^2)q_{it},
    \qquad
    y_v\le(1+\eta^2)q_v
\]
for every state $(i,t)$ and offline vertex $v$;
\item[(c)] \emph{locality:} the free-status queries
$u\mapsto\mathbf1\{u\notin V(M_C)\}$, over arrivals and offline
vertices alike, are answered by a natural local rule in the probe model
of \Cref{def:prelim-access,def:prelim-rule}, with out-query and
expected in-query complexity at most
\[
    \Lambda_\eta(\tau)^{\widehat a_\eta},
    \qquad
    \widehat a_\eta
    :=
    C_{\rm loc}\eta^{-2}\log^2\frac e\eta .
\]
\end{itemize}
The rule may depend on the fixed instance, capped-benchmark marginals,
and parameters $\eta,\tau$.
\end{lemma}

\begin{proof}
Set
\[
    \alpha
    :=
    \frac{\eta}{4(1+c_{\rm lca})},
    \qquad
    \delta:=\eta^2,
    \qquad
    \gamma:=c_{\rm lca}\alpha.
\]
Thus $\alpha\le\eta$, $\gamma\ge c_{\rm lca}\alpha$, and
\[
    2\gamma
    =
    \frac{c_{\rm lca}}{2(1+c_{\rm lca})}\eta
    <1.
\]
We will invoke \Cref{thm:prelim-matching} with accuracy parameter
$\alpha$.

If $Q=0$, then \eqref{eq:new-cap-value} gives $x(C_\tau)=0$, and
$M_C=\varnothing$ satisfies the lemma.  Assume henceforth that $Q>0$.
The instance, capped-benchmark marginals, and parameters are fixed
throughout the proof.

Let
\[
    \mathcal U
    :=
    \{(i,t):t\in\mathcal T_i\}\sqcup V
\]
be the endpoint index set: an online endpoint is indexed by a state
$(i,t)$ and an offline endpoint by a vertex $v$.
Thus $q_u$ and, below, $y_u$ denote unconditional endpoint marginals
on either side.

Write
\[
    \mathcal U_+:=\{u\in\mathcal U:q_u>0\}.
\]
We work on the endpoint-marginal feasibility box
\[
    K
    :=
    \prod_{u\in\mathcal U_+}[0,\bar y_u],
    \qquad
    \bar y_{it}=p_i(t),\quad \bar y_v=1,
\]
and extend vectors in $K$ by zero outside $\mathcal U_+$.

We preview the construction of the crucial witness.  Given a candidate
endpoint-marginal vector $y\in K$, we use $y$ to assign scores to the
edges of each realization of $\mathcal G_C^D$, and then use independent
random thresholds to retain a weighted subgraph.  Denote by $M(y)$ the
random matching obtained by running the local weighted-matching rule
$\Psi$ of \Cref{thm:prelim-matching} on this weighted graph.
Let $\mathcal R(y)$ be the vector of unconditional endpoint marginals
of $M(y)$, averaged over the arrival realization and all auxiliary
randomness.

Thus $\mathcal R$ maps a candidate marginal vector $y$ to the marginals
actually produced by the matching rule in response to the weights
induced by $y$.  We will show that $\mathcal R$ is a continuous
self-map of $K$.  Since $K$ is nonempty, compact, and convex, Brouwer's
theorem then gives a fixed point
\[
    y^\star=\mathcal R(y^\star).
\]
We define the crucial witness by $M_C:=M(y^\star)$.

\paragraph{Penalized response map.}
We restrict the response graph to state-edges whose two endpoint
indices lie in $\mathcal U_+$.  This removes no edge of
$M_C^{\star,D}$ with positive probability.

For $y\in K$, define
\[
    \operatorname{val}(y)
    :=
    \frac12\sum_{u\in\mathcal U_+}y_u.
\]
When $y$ is the endpoint-marginal vector of a random matching,
$\operatorname{val}(y)$ is its expected size.

Define
\[
    h_\delta(z)
    :=
    \begin{cases}
        0,
            & z\le0,\\[1mm]
        \dfrac{z^2}{2\delta},
            & 0<z<\delta,\\[2mm]
        z-\dfrac{\delta}{2},
            & z\ge\delta.
    \end{cases}
\]
This function is convex and continuously differentiable, with
\[
    h_\delta'(z)
    =
    \begin{cases}
        0,
            & z\le0,\\[1mm]
        \dfrac{z}{\delta},
            & 0<z<\delta,\\[2mm]
        1,
            & z\ge\delta.
    \end{cases}
\]
In particular, $0\le h_\delta'(z)\le1$.

Define the one-sided overuse penalty
\[
    P(y)
    :=
    \sum_{u\in\mathcal U_+}
        q_u
        h_\delta\!\left(\frac{y_u-q_u}{q_u}\right),
\]
and the penalized objective
\[
    \Phi(y):=\operatorname{val}(y)-P(y).
\]
The endpoint price is
\begin{equation}
\label{eq:endpointprice}
    s_u(y)
    :=
    \frac{\partial P}{\partial y_u}(y)
    =
    h_\delta'\!\left(\frac{y_u-q_u}{q_u}\right),
\end{equation}
so
\[
    0\le s_u(y)\le1.
\]
Assign an edge with endpoints $u,v$ the score
\[
    g_{uv}(y):=1-s_u(y)-s_v(y).
\]
Thus $g_{uv}(y)\le1$, although it may be negative.

For every potential crucial edge $e$, independently draw a locally
owned threshold
\[
    \Theta_e\sim\operatorname{Unif}[\gamma,2\gamma].
\]
For $y\in K$ and a realization of the arrivals and thresholds, let
$G_y$ be the subgraph of $\mathcal G_C^D$ obtained by retaining exactly
those edges $e$ for which
\[
    g_e(y)\ge\Theta_e.
\]
Every edge of $G_y$ therefore has score in $[\gamma,1]$.
On $G_y$, with edge weights $g_e(y)$, apply the local weighted-matching
rule $\Psi$ of \Cref{thm:prelim-matching}, with accuracy parameter
$\alpha$.  Let
\[
    M(y):=\Psi(G_y,\pi_\Psi;g(y)),
\]
where $\pi_\Psi$ denotes the internal randomness of $\Psi$.
The threshold marks $\Theta$ and $\pi_\Psi$ are mutually independent
and independent of the arrivals and menu randomness.

Define the response map $\mathcal R$ to be the unconditional
endpoint-marginal vector of $M(y)$:
\[
    \mathcal R_{it}(y)
    :=
    \Pr[T_i=t,\ i\in V(M(y))],
    \qquad
    \mathcal R_v(y)
    :=
    \Pr[v\in V(M(y))].
\]
The probabilities average over the arrivals, threshold marks, and
internal randomness of $\Psi$, so $\mathcal R$ is deterministic once
the instance and capped-benchmark data are fixed.

We verify the hypotheses of Brouwer's fixed-point theorem.
First, $K$ is nonempty, compact, and convex.  Second,
$\mathcal R(K)\subseteq K$: every output of $\Psi$ is a matching, so
\[
    \mathcal R_{it}(y)\le p_i(t),
    \qquad
    \mathcal R_v(y)\le1.
\]

It remains to prove continuity.  Fix $y^0\in K$.
The prices $s_u(y)$, and hence the edge scores $g_e(y)$, are continuous
in $y$.  Since there are finitely many potential edges and every
$\Theta_e$ has a continuous distribution, with probability one
\[
    \Theta_e\ne g_e(y^0)
    \qquad\text{for every potential edge }e.
\]
For every realization outside this null event, the thresholded graph
$G_y$ is therefore constant throughout some neighborhood of $y^0$.
On this fixed graph the surviving weights $g_e(y)$ vary continuously
with $y$ and remain in $[\gamma,1]$.
Part~(ii) of \Cref{thm:prelim-matching} implies that the inclusion
probability of every edge in the output matching is continuous in
these weights.  Endpoint marginals are finite sums of such
edge-inclusion probabilities and are therefore continuous as well.
Since they are bounded by one, dominated convergence over the arrivals
and threshold marks shows that the unconditional response map
$\mathcal R$ is continuous.

Thus $\mathcal R$ is a continuous self-map of the nonempty compact
convex set $K$.  Brouwer's theorem gives a fixed point
\[
    y^\star=\mathcal R(y^\star).
\]
Fix one such $y^\star$ and define
\[
    M_C:=M(y^\star).
\]
The fixed point $y^\star$ is deterministic; all randomness of $M_C$
comes from the arrivals and the auxiliary dictionary
\[
    \pi:=(\Theta,\pi_\Psi).
\]

By definition, $\mathcal R(y^\star)$ is exactly the endpoint-marginal
vector of $M_C$.  Since $y^\star=\mathcal R(y^\star)$, the coordinates
of $y^\star$ are therefore precisely the endpoint marginals appearing
in the statement of the lemma.  Consequently,
\begin{equation}
\label{eq:new-fixed-point-value}
    \operatorname{val}(y^\star)
    =
    \frac12\sum_{u\in\mathcal U_+}y^\star_u
    =
    \mathbb E[|M_C|],
\end{equation}
because every matching edge contributes exactly two matched endpoints.

\paragraph{Endpoint-load control.}
Suppose that, for some $u\in\mathcal U_+$,
\[
    y^\star_u\ge(1+\delta)q_u.
\]
Then \eqref{eq:endpointprice} gives $s_u(y^\star)=1$.
Since every endpoint price is nonnegative, every edge $e$ incident to
$u$ has
\[
    g_e(y^\star)\le0.
\]
But every threshold satisfies $\Theta_e\ge\gamma>0$, so no edge
incident to $u$ survives in $G_{y^\star}$.  Hence $u$ is never matched
by $M(y^\star)$ and
\[
    \mathcal R_u(y^\star)=0,
\]
contradicting
\[
    y^\star_u=\mathcal R_u(y^\star)
    \ge(1+\delta)q_u>0.
\]
Thus
\[
    y^\star_u<(1+\delta)q_u
    \qquad\text{for every }u\in\mathcal U_+.
\]
For $u\notin\mathcal U_+$, both sides are zero by construction.
Since $\delta=\eta^2$,
\begin{equation}
\label{eq:pointwise-load}
    y^\star_u\le(1+\eta^2)q_u
    \qquad\text{for every endpoint }u.
\end{equation}
This proves part~(b).

Summing \eqref{eq:pointwise-load} over all endpoints and using
$\sum_uq_u=2Q$ gives
\begin{equation}
\label{eq:witness-upper-value}
    \mathbb E[|M_C|]
    =
    \operatorname{val}(y^\star)
    \le
    (1+\eta^2)Q.
\end{equation}

\paragraph{Value.}
For this paragraph, write
\[
    g_e:=g_e(y^\star),
    \qquad
    g(M):=\sum_{e\in M}g_e.
\]
Since $P$ is convex, $\Phi$ is concave, with
\[
    \frac{\partial\Phi}{\partial y_u}(y)
    =
    \frac12-s_u(y).
\]
The capped-benchmark endpoint vector $q$ belongs to $K$, and
\[
    P(q)=0,
    \qquad
    \operatorname{val}(q)
    =
    \frac12\sum_uq_u
    =
    Q.
\]
Hence $\Phi(q)=Q$.  Applying the tangent inequality for the concave
function $\Phi$ at $y^\star$, evaluated at $q$, gives
\[
    Q-\Phi(y^\star)
    \le
    \sum_u
        \left(\frac12-s_u(y^\star)\right)
        (q_u-y^\star_u).
\]

For any random matching $M$ supported on endpoint indices in
$\mathcal U_+$, with endpoint marginals $r_u$, its expected score at
the prices determined by $y^\star$ is
\begin{equation}
\label{eq:score-from-marginals}
    \mathbb E[g(M)]
    =
    \sum_u
        \left(\frac12-s_u(y^\star)\right)r_u.
\end{equation}
Applying \eqref{eq:score-from-marginals} to
$M_C^{\star,D}$ and $M_C$ gives
\[
    \mathbb E[g(M_C^{\star,D})]
    =
    \sum_u
        \left(\frac12-s_u(y^\star)\right)q_u
\]
and
\[
    \mathbb E[g(M_C)]
    =
    \sum_u
        \left(\frac12-s_u(y^\star)\right)y^\star_u.
\]
Therefore
\begin{equation}
\label{eq:new-tangent}
    Q-\Phi(y^\star)
    \le
    \mathbb E[g(M_C^{\star,D})]
    -
    \mathbb E[g(M_C)].
\end{equation}

For a realization of the arrivals and threshold marks, let $W^\star$
be the maximum matching weight of $G_{y^\star}$ under the score
weights $g_e$.
For each realization of the arrivals, benchmark randomness, and
threshold marks, the retained edges of $M_C^{\star,D}$ form a feasible
matching in $G_{y^\star}$ and hence have total score at most $W^\star$.
Every benchmark edge removed by thresholding has score below
$2\gamma$.  Therefore, pointwise,
\[
    g(M_C^{\star,D})
    \le
    W^\star+2\gamma|M_C^{\star,D}|.
\]
Taking expectations and using
$\mathbb E[|M_C^{\star,D}|]=Q$ gives
\begin{equation}
\label{eq:new-benchmark-score}
    \mathbb E[g(M_C^{\star,D})]
    \le
    \mathbb E[W^\star]+2\gamma Q.
\end{equation}

Conditional on the arrivals and threshold marks, $G_{y^\star}$ and
its score weights are fixed and lie in $[\gamma,1]$.
Applying part~(i) of \Cref{thm:prelim-matching} with accuracy
parameter $\alpha$ gives
\[
    \mathbb E_{\pi_\Psi}
        [g(M_C)\mid T,\Theta]
    \ge
    (1-\alpha)W^\star.
\]
Averaging yields
\begin{equation}
\label{eq:new-response-score}
    \mathbb E[g(M_C)]
    \ge
    (1-\alpha)\mathbb E[W^\star].
\end{equation}

Every edge used by $M_C$ has score at most $1$, so
\[
    \mathbb E[g(M_C)]
    \le
    \mathbb E[|M_C|]
    \le
    (1+\eta^2)Q
\]
by \eqref{eq:witness-upper-value}.  Hence
\[
    \mathbb E[W^\star]
    \le
    \frac{1+\eta^2}{1-\alpha}Q
    \le
    2Q,
\]
where the last inequality uses $\eta\le1/4$.

Substituting \eqref{eq:new-benchmark-score} and
\eqref{eq:new-response-score} into \eqref{eq:new-tangent} gives
\[
\begin{aligned}
    Q-\Phi(y^\star)
    &\le
    \alpha\,\mathbb E[W^\star]+2\gamma Q \\
    &\le
    2\alpha Q+2c_{\rm lca}\alpha Q \\
    &=
    2(1+c_{\rm lca})\alpha Q
    =
    \frac{\eta}{2}Q.
\end{aligned}
\]
Thus
\[
    \Phi(y^\star)\ge(1-\eta/2)Q.
\]
Since $P(y^\star)\ge0$, \eqref{eq:new-fixed-point-value} gives
\[
    \mathbb E[|M_C|]
    =
    \operatorname{val}(y^\star)
    \ge
    \Phi(y^\star)
    \ge
    (1-\eta)Q,
\]
proving part~(a).

Combining this with
$Q\ge(1-\eta)x(C_\tau)$ from \Cref{lem:capping} also gives
\begin{equation}
\label{eq:new-crucial-value}
    \mathbb E[|M_C|]
    \ge
    (1-2\eta)x(C_\tau).
\end{equation}

\paragraph{Locality.}
The witness is computed from the capped graph using the locally owned
threshold marks $\Theta$ and the internal dictionary $\pi_\Psi$, all
independent of the arrivals and menu randomness.  The fixed point
$y^\star$ and the capped-benchmark marginals are deterministic
parameters of the rule.

The thresholded graph has maximum degree at most
\[
    D=\Lambda_\eta(\tau).
\]
The weighted-matching rule is invoked with accuracy $\alpha$, so
\Cref{thm:prelim-matching} gives query profile
\[
    D^{a_\alpha},
    \qquad
    a_\alpha
    =
    C_{\rm lca}\alpha^{-2}\log^2\frac e\alpha.
\]
Simulating the threshold-filtering layer from the capped graph costs
at most one additional factor $D$.  Since $\alpha$ is a fixed
universal constant multiple of $\eta$, choosing $C_{\rm loc}$
sufficiently large gives
\[
    a_\alpha+1
    \le
    \widehat a_\eta
    =
    C_{\rm loc}\eta^{-2}\log^2\frac e\eta .
\]
Consequently
\begin{equation}
\label{eq:new-lca-complexity}
    R^+,R^-
    \le
    \Lambda_\eta(\tau)^{\widehat a_\eta}.
\end{equation}
This proves part~(c) and completes the proof.
\end{proof}

\paragraph{Presence of the crucial witness.}
Every edge of $M_C$ is crucial.  Conditional on the arrival profile
and on $M_C$, the menu randomness is independent of the witness, so
(D1) gives
\[
    \Pr[e\notin H\mid T,M_C]
    \le
    e^{-k\tau}
    \qquad
    \text{for every }e\in M_C.
\]
Hence, whenever $k\tau\ge\log(1/\eta)$,
\[
    \mathbb E[|M_C\cap H|\mid T,M_C]
    \ge
    (1-\eta)|M_C|.
\]
Averaging and applying \Cref{lem:witness},
\begin{equation}
\label{eq:witness-presence}
    \mathbb E[|M_C\cap H|]
    \ge
    (1-\eta)\mathbb E[|M_C|]
    \ge
    (1-\eta)^2Q.
\end{equation}
Together with $Q\ge(1-\eta)x(C_\tau)$, this also gives
\[
    \mathbb E[|M_C\cap H|]
    \ge
    (1-3\eta)x(C_\tau).
\]
For VarOpt menus, the same threshold condition implies
$1/k\le\tau$, so every crucial edge, and hence every edge of $M_C$,
is retained deterministically.

\paragraph{Interface to light completion.}
From this point on, we use only four properties of $M_C$: its
near-optimal value, the pointwise endpoint-load control
\[
    y_u\le(1+\eta^2)q_u,
\]
independence from the menus conditional on the arrival profile, and
the free-status query profile \eqref{eq:new-lca-complexity}.
The pointwise load bound lets us transfer the benchmark light plan to
the marginal endpoint capacities left by $M_C$ with only
$O(\eta Z)$ loss, while the independence and locality properties
provide the variance and bias bounds needed to recover that mass from
sampled light edges.

\section{Fitting the light benchmark to the witness}
\label{subsec:light-fit}

Fix $0<\tau^-<\tau$ and recall the decomposition
$C_\tau,I_{\tau^-,\tau},L_{\tau^-}$ from \Cref{eq:taubands}.  Let
\[
    a_{it}:=1-\frac{y_{it}}{p_i(t)}
    \quad (p_i(t)>0),
    \qquad
    a_v:=1-y_v
\]
be the marginal endpoint availabilities left by $M_C$.

The benchmark light plan is feasible in the residual capacities left
by the capped benchmark $M_C^{\star,D}$, whereas we need a plan
feasible in the capacities left by $M_C$.  The pointwise load control
of \Cref{lem:witness} shows that these two sets of capacities are close,
except possibly at endpoints with very small residual availability.
We discard such endpoints and then use one common rescaling.

\begin{lemma}[Residual fitting]
\label{lem:residual-fitting}
There is a deterministic subplan $x^0$ supported on
$L_{\tau^-}$ such that
\begin{equation}
\label{eq:new-x0}
    |x^0|\ge x(L_{\tau^-})-C\eta Z,
\end{equation}
and, writing $z^0_{itv}:=x^0_{itv}/p_i(t)$,
\[
    0\le z^0_{itv}\le\tau^-,
    \qquad
    \sum_vz^0_{itv}\le a_{it},
    \qquad
    \sum_{i,t}x^0_{itv}\le a_v.
\]
Every endpoint in the support of $x^0$ has residual availability at
least $\eta$.
\end{lemma}

\begin{proof}
For an endpoint $u$, write $b_u=p_i(t)$ if $u=(i,t)$ and $b_u=1$
if $u=v$, so $b_u-y_u=a_ub_u$.  Let $\ell_u$ be the load of the
benchmark light plan at $u$.  The light benchmark and the capped
crucial benchmark are disjoint parts of the same matching, so
\begin{equation}
\label{eq:light-benchmark-residual}
    \ell_u+q_u\le b_u.
\end{equation}

Delete all light edges incident to endpoints with $a_u<\eta$.
At each such endpoint, the witness bound gives
\[
    (1-\eta)b_u<y_u\le(1+\eta^2)q_u,
    \qquad
    \ell_u\le b_u-q_u
    \le\frac{\eta+\eta^2}{1-\eta}\,q_u.
\]
Since the $q$-loads sum to $Q$ on each side, the total deleted mass
is at most $2(\eta+\eta^2)Q/(1-\eta)=O(\eta Z)$.

Let $\bar x$ be the remaining plan.  On its support,
$b_u-y_u\ge\eta b_u$, and hence
\[
\begin{aligned}
    b_u-q_u
    &=(b_u-y_u)+(y_u-q_u)\\
    &\le(b_u-y_u)+\eta^2b_u
     \le(1+\eta)(b_u-y_u).
\end{aligned}
\]
Thus $x^0:=\bar x/(1+\eta)$ fits every residual marginal capacity
$b_u-y_u$.  This scaling loses at most $\eta Z$ additional mass.
Dividing the online-state constraints by $p_i(t)$ gives the claimed
row bounds; the offline constraints give the column bounds.
Finally, deletion and scaling preserve
$0\le z^0_{itv}\le z_{itv}\le\tau^-$, and every remaining endpoint
has availability at least $\eta$.
\end{proof}

\section{Sampling the light residual}
\label{subsec:light-sampling}

The fitted light plan $x^0$ from \Cref{lem:residual-fitting}
satisfies the marginal capacity constraints left by $M_C$.
We now recover its mass, up to $O(\eta Z)$ expected loss, as a
feasible fractional matching $\widehat Y$ on retained light edges
whose endpoints are free in $M_C$.  Combining $\widehat Y$ with
$M_C\cap H$ completes the fractional matching certificate
supported on $H$.

We begin with moment bounds for the menu weights.
We then use the locality of $M_C$ to bound variances and
conditioning bias for weighted sums of free-status indicators.
These estimates allow us to control the mass lost when making
the light-edge loads feasible.

\paragraph{Menu weights.}
The analysis of the given derived menu rule uses only (D1)
and the nonpositive covariances in (D2).
For a state-edge in the support of
$x^0$, the Horvitz--Thompson weight
\[
    W_{iv}
    :=
    z^0_{itv}\,
    \frac{\mathbf1\{v\in S_i\}}
         {\Pr[v\in S_i\mid T_i=t]}
\]
has conditional mean $z^0_{itv}$.  Whenever $1/k\le\tau^-$,
using \eqref{eq:inclusion-lower-bound}, $z^0_{itv}\le z_{itv}$,
and $\max\{z^0_{itv},1/k\}\le\tau^-$, we obtain
\[
    \mathbb E[W_{iv}^2\mid T_i=t]
    =
    \frac{(z^0_{itv})^2}{\Pr[v\in S_i\mid T_i=t]}
    \le
    \frac85 z^0_{itv}\max\{z^0_{itv},1/k\}
    \le
    4\tau^- z^0_{itv}.
\]
Conditional on $T_i=t$, distinct weights from the same menu have
nonpositive covariance.  Conditional on the arrival profile $T$, menus of different arrivals are independent and are independent of $M_C$.

\subsection{Decorrelation from locality.}
The fixed-input correlation bound
(\Cref{lem:prelim-fixed-correlation}) controls dependence among
free-status indicators conditional on the arrival states.
The following theorem uses capping stability to account also
for the arrival randomness, bounding variances of weighted sums
and the bias caused by conditioning on a single arrival state.

\begin{restatable}[In--out decorrelation transfer]{theorem}{decorrelationtransfer}
\label{thm:prelim-transfer}
Let $M_C$ be a matching of $\mathcal G_C^D(T)$ computed by a
natural local rule whose free-status queries have query profile
$(R^+,R^-)$ for every arrival realization $T$.
The input consists of the capped graph with each arrival vertex
$i$ labelled by $T_i$, together with deterministic instance data.
The label $T_i$ is revealed only when $i$ is probed.
The shared dictionary $\pi$ is independent of $T$ and of the
menu randomness.  Put
\[
    F_v:=\mathbf1\{v\notin V(M_C)\},
    \qquad
    F_i:=\mathbf1\{i\notin V(M_C)\}.
\]
There is a universal constant $C_{\mathrm{dec}}\ge1$ such that,
with
\[
    \Gamma
    :=
    C_{\mathrm{dec}}
    \bigl(1+d^2D(D+d+1)\bigr)R^+R^-,
\]
the following bounds hold for all deterministic coefficient
families $(h_v)_v$ and $(h_{it})_{i,t}$:
\[
    \Var\!\left(\sum_v h_vF_v\right)
    \le \Gamma\sum_v h_v^2,
    \qquad
    \Var\!\left(
        \sum_{i,t}h_{it}\mathbf1\{T_i=t\}F_i
    \right)
    \le \Gamma\sum_{i,t}p_i(t)h_{it}^2.
\]
The first inequality remains valid after conditioning on any
single state $T_i=t$ with $p_i(t)>0$.
In addition, for every deterministic coefficient array satisfying
$0\le h_{itv}\le h_{\max}$,
\[
    \sum_{i,t}p_i(t)
    \left|
        \sum_v h_{itv}
        \bigl(
            \Pr[F_v=1\mid T_i=t]-\Pr[F_v=1]
        \bigr)
    \right|
    \le
    C_{\mathrm{dec}}d^2DR^-h_{\max}x(C_\tau).
\]
\end{restatable}

The proof is given in \Cref{app:decorrelation}.
We apply the theorem in \Cref{lem:light-completion}.
Together with the menu moment bounds above, its variance bounds
control fluctuations in the sampled light-edge loads, while its
bias bound controls the error caused by conditioning on an
arrival state.

For the witness of \Cref{lem:witness},
\[
    R^+,R^-\le\Lambda_\eta(\tau)^{\widehat a_\eta},
    \qquad
    d\le D=\Lambda_\eta(\tau).
\]
Consequently, choosing a sufficiently large universal constant
$C_4$ and defining
\begin{equation}
\label{eq:new-Aeta}
    A_\eta
    :=
    C_4\eta^{-2}\log^2\frac{C_4}{\eta},
\end{equation}
we obtain
\begin{equation}
\label{eq:new-Gamma}
    \Gamma\le\Lambda_\eta(\tau)^{A_\eta}.
\end{equation}

\subsection{Light completion}
\label{app:light-completion}

We now combine the menu moment bounds with
\Cref{thm:prelim-transfer} to construct the light completion.
For this construction, we require
\begin{equation}
\label{eq:new-b-small}
    4\tau^{-}
    \le
    \eta^{C_{\rm light}}
    \Lambda_\eta(\tau)^{-(A_\eta+C_{\rm light})},
\end{equation}
where $C_{\rm light}\ge6$ is a universal constant.
We choose thresholds satisfying this condition in
\Cref{sec:thresholds}.

\begin{lemma}[Light completion]
\label{lem:light-completion}
Assume $1/k\le\tau^{-}$ and \eqref{eq:new-b-small}.  Let $x^0$ be the
deterministic light subplan given by \Cref{lem:residual-fitting}, so that
\[
    0\le z^0_{itv}:=\frac{x^0_{itv}}{p_i(t)}\le\tau^{-},
    \qquad
    \sum_vz^0_{itv}\le a_{it},
    \qquad
    \sum_{i,t}x^0_{itv}\le a_v,
\]
and every endpoint in its support has residual availability at least
$\eta$.  Then there is a random feasible fractional matching
$\widehat Y$, supported on retained edges from the support of $x^0$ and
using only vertices free in $M_C$, such that
\[
    \mathbb E[|\widehat Y|]\ge |x^0|-3\eta Z.
\]
\end{lemma}

\begin{proof}
Set
\[
    \ell_{it}:=\sum_vz^0_{itv}\le a_{it},
    \qquad
    \ell_v:=\sum_{i,t}x^0_{itv}\le a_v.
\]
Both $\sum_{i,t}p_i(t)\ell_{it}$ and $\sum_v\ell_v$ equal
$|x^0|\le Z$.  Throughout the proof, sums and quotients involving
availability denominators are restricted to the support of $x^0$;
there all availabilities are at least $\eta$.  All other edge weights
and endpoint loads are zero.

\paragraph{Preload moments.}
Recall the menu weights $W_{iv}$ defined above and set $b:=4\tau^-$.
Their moment bounds are
\[
    \mathbb E[W_{iv}\mid T_i=t]=z^0_{itv},
    \qquad
    \mathbb E[W_{iv}^2\mid T_i=t]\le b z^0_{itv}.
\]
Conditional on $T_i=t$, distinct weights from the same menu have
nonpositive covariance by (D2), and the menu of $i$ is independent of
the witness.  Define the preloads
\[
    P_i:=\sum_vW_{iv}\frac{F_v}{a_v},
    \qquad
    P_v:=\sum_iW_{iv}\frac{F_i}{a_{iT_i}}.
\]

For a row, condition first on $T_i=t$ and its menu.  The conditional
offline variance bound of \Cref{thm:prelim-transfer} gives
\[
    \Var(P_i\mid T_i=t,\mathrm{menu}_i)
    \le\Gamma\sum_v\frac{W_{iv}^2}{a_v^2}.
\]
The expectation of this bound over the menu is at most
$\Gamma b\ell_{it}/\eta^2$.
The conditional mean is
$\sum_vW_{iv}\Pr[F_v=1\mid T_i=t]/a_v$; its menu variance is at
most $b\ell_{it}/\eta^2$ by (D2).

For a column, independently pre-sample a latent menu $S_{it}$ for
every state $(i,t)$, independently of $T$ and the witness dictionary,
and use $S_{iT_i}$ as the realized menu.  Let $W_{itv}$ denote its
latent Horvitz--Thompson weight.  Conditional on this table,
\[
    P_v=\sum_{i,t}\frac{W_{itv}}{a_{it}}
                  \mathbf1\{T_i=t\}F_i,
    \qquad
    \mathbb E[P_v\mid\mathrm{menu\ table}]
       =\sum_{i,t}p_i(t)W_{itv}.
\]
The mean identity uses
$\Pr[F_i=1\mid T_i=t]=a_{it}$ and independence of the table from
$T$ and the witness.  The online-state variance bound of
\Cref{thm:prelim-transfer}, averaged over the table, is at most
$\Gamma b\ell_v/\eta^2$.
Independence of the latent menus bounds the variance of the displayed
conditional mean by
$\sum_{i,t}p_i(t)^2\mathbb E[W_{itv}^2]\le b\ell_v$.
Thus the law of total variance, $\Gamma\ge1$, and $\eta\le1$ give
\begin{equation}
\label{eq:new-preload-var}
    \Var(P_i\mid T_i=t)\le\xi\ell_{it},
    \qquad
    \Var(P_v)\le\xi\ell_v,
    \qquad
    \xi:=\frac{2\Gamma b}{\eta^2}.
\end{equation}
The column mean is exact: $\mathbb E[P_v]=\ell_v$.
For the row means, write
\[
    \delta_{it}:=
       \left|\mathbb E[P_i\mid T_i=t]-\ell_{it}\right|,
    \qquad
    B_{\mathrm{bias}}:=\sum_{i,t}p_i(t)\delta_{it}.
\]
Applying the bias bound of \Cref{thm:prelim-transfer} with
$h_{itv}=z^0_{itv}/a_v\le\tau^-/\eta$ gives
\begin{equation}
\label{eq:light-aggregate-bias}
    B_{\mathrm{bias}}
    \le C_{\mathrm{dec}}d^2DR^-\frac{\tau^-}{\eta}x(C_\tau)
    \le\frac{\Gamma b}{4\eta}Z
    =\frac{\xi\eta}{8}Z.
\end{equation}
Here the definition of $\Gamma$ and $R^+\ge1$ justify the second
inequality.  By \eqref{eq:new-Gamma}, \eqref{eq:new-b-small}, and
$C_{\rm light}\ge6$, using $D=\Lambda_\eta(\tau)\ge2$, we have
\begin{equation}
\label{eq:light-effective-smallness}
    \Gamma b\le\frac{\eta^6}{2},
    \qquad
    \xi\le\eta^4,
    \qquad
    B_{\mathrm{bias}}\le\frac{\eta^5}{8}Z.
\end{equation}

\paragraph{Value before removing overload.}
Define the availability-normalized weights
\[
    U_{iv}:=W_{iv}\frac{F_iF_v}{a_{iT_i}a_v}.
\]
Reading their value by columns and using the exact column means gives
\[
    \mathbb E[|U|]-|x^0|
       =\sum_v\frac{\Cov(F_v,P_v)}{a_v}.
\]
Since $\Var(F_v)=a_v(1-a_v)=a_vy_v$, Cauchy--Schwarz gives
\[
\begin{aligned}
    \left|\mathbb E[|U|]-|x^0|\right|
    &\le\sum_v\sqrt{\frac{\xi\ell_vy_v}{a_v}}\\
    &\le\sqrt{\frac{\xi}{\eta}}
       \sqrt{\left(\sum_v\ell_v\right)\left(\sum_vy_v\right)}\\
    &\le\sqrt{\frac{(1+\eta^2)\xi}{\eta}}\,Z
     \le\eta Z.
\end{aligned}
\]
We used $\sum_vy_v\le(1+\eta^2)Q\le(1+\eta^2)Z$ from
\Cref{lem:witness}, followed by
\eqref{eq:light-effective-smallness} and $\eta\le1/4$.
Now create slack by setting
\[
    \widetilde Y_{iv}:=\frac{U_{iv}}{1+\eta},
    \qquad
    R_i:=\sum_v\widetilde Y_{iv},
    \qquad
    L_v:=\sum_i\widetilde Y_{iv}.
\]
Then
\begin{equation}
\label{eq:new-Y-value}
    \mathbb E[|\widetilde Y|]\ge|x^0|-2\eta Z.
\end{equation}

\paragraph{Direct control of overload.}
We use a scalar estimate that requires no independence between a
preload and its free-status indicator.  Let $F$ be Bernoulli with
mean $a>0$, let $P$ have mean $\mu$ and variance $\sigma^2$, and
let $\ell\le a$.  With $D_P:=P-\mu$,
\[
\begin{aligned}
    \left(\frac{FP}{(1+\eta)a}-1\right)_+
    &=\frac{F}{(1+\eta)a}\bigl(P-(1+\eta)a\bigr)_+\\
    &\le\frac{F}{(1+\eta)a}
       \bigl((D_P-\eta a)_++(\mu-\ell)_+\bigr).
\end{aligned}
\]
Using $(x-s)_+\le x^2/(4s)$ for $s>0$, $F\le1$, and
$\mathbb E[F]=a$, we obtain
\begin{equation}
\label{eq:new-overload-scalar}
    \mathbb E\!\left[
       \left(\frac{FP}{(1+\eta)a}-1\right)_+
    \right]
    \le(\mu-\ell)_++\frac{\sigma^2}{4\eta a^2}.
\end{equation}

For a row, apply this conditionally on $T_i=t$, with
$F=F_i$, $P=P_i$, $a=a_{it}$, and $\ell=\ell_{it}$.
For a column, apply it unconditionally with
$F=F_v$, $P=P_v$, $a=a_v$, and $\ell=\ell_v$; its mean error is
zero.  Summing over row states with weights $p_i(t)$ and over
columns, \eqref{eq:new-preload-var} and
\eqref{eq:light-aggregate-bias} give
\[
\begin{aligned}
    \mathbb E\!\left[
       \sum_i(R_i-1)_++\sum_v(L_v-1)_+
    \right]
    &\le B_{\mathrm{bias}}
       +\frac{\xi}{2\eta^3}|x^0|\\
    &\le\left(\frac{\eta^5}{8}+\frac\eta2\right)Z
     \le\eta Z.
\end{aligned}
\]

In each realization, scale every overloaded row of $\widetilde Y$
down to load one, then do the same for the remaining overloaded
columns.  The first step loses exactly $\sum_i(R_i-1)_+$ and can
only reduce column loads, so the second loses at most
$\sum_v(L_v-1)_+$.  Column scaling preserves row feasibility.
The resulting $\widehat Y$ is feasible, is supported on retained
light edges, and uses only vertices free in $M_C$.
Together with \eqref{eq:new-Y-value}, this gives
\[
    \mathbb E[|\widehat Y|]\ge|x^0|-3\eta Z.
\]
\end{proof}

\section{Choosing thresholds and completing the proof}
\label{sec:thresholds}

We now choose $\tau^-<\tau$ so that \eqref{eq:new-b-small} holds and
the intermediate band has small benchmark mass.  Let
\[
    m:=\lceil1/\eta\rceil,
    \qquad
    C_3:=C_{\rm light}+1,
    \qquad
    B_\eta:=A_\eta+C_3.
\]  Starting from $\rho_0=1/4$, define
\begin{equation}
\label{eq:new-rho-rec}
    \rho_{j+1}
    :=
    \tfrac12\eta^{C_3}\Lambda_\eta(\rho_j)^{-B_\eta},
    \qquad 0\le j<m.
\end{equation}
Since $\Lambda_\eta(\rho)\ge2e/\rho$ and $B_\eta\ge1$, we have
\[
    \rho_{j+1}
    \le
    \frac{\eta^{C_3}}{4e}\rho_j
    \le
    \eta\rho_j.
\]
Thus the bands
\[
    \mathcal B_j
    :=
    \{(i,t,v):\rho_{j+1}<z_{itv}<\rho_j\},
    \qquad 0\le j<m,
\]
are disjoint.  Their total benchmark mass is at most $Z$, so for some
$j<m$,
\[
    x(\mathcal B_j)\le Z/m\le\eta Z.
\]
Fix such a $j$ and set
\[
    \tau:=\rho_j,
    \qquad
    \tau^-:=\rho_{j+1}.
\]
Then
\[
    4\tau^-
    =
    2\eta^{C_3}\Lambda_\eta(\tau)^{-B_\eta}
    \le
    \eta^{C_{\rm light}}
    \Lambda_\eta(\tau)^{-(A_\eta+C_{\rm light})},
\]
where the inequality uses $2\eta/\Lambda_\eta(\tau)\le1$.
Thus \eqref{eq:new-b-small} holds.  Moreover
$\rho_{j+1}\le\eta\rho_j$ gives
\begin{equation}
\label{eq:threshold-separation}
    \frac{\tau}{\tau^-}\ge\frac1\eta
    \ge\log\frac1\eta.
\end{equation}
Hence, whenever $1/k\le\tau^-$, we also have
$k\tau\ge\log(1/\eta)$, as required for
\eqref{eq:witness-presence}.

The complete certificate is the fractional matching
\[
    Y:=\mathbf1_{M_C\cap H}+\widehat Y,
\]
where $\mathbf1_{M_C\cap H}$ is the edge-incidence vector of the
retained crucial witness.  Since $\widehat Y$ uses only retained light
edges and vertices free in $M_C$, $Y$ is feasible and supported on $H$.
Bipartite matching integrality gives $\nu(H)\ge |Y|$ in every
realization.  Using
\Cref{lem:witness,lem:residual-fitting,lem:light-completion} and
\eqref{eq:witness-presence},
\[
\begin{aligned}
    \mathbb E[\nu(H)]
    &\ge
    \mathbb E[|M_C\cap H|]+\mathbb E[|\widehat Y|] \\
    &\ge
    x(C_\tau)+x(L_{\tau^-})-C\eta Z \\
    &=
    Z-x(I_{\tau^-,\tau})-C\eta Z \\
    &\ge
    (1-C'\eta)Z.
\end{aligned}
\]
Choosing the constant $c_0$ in $\eta=c_0\varepsilon$ so that
$C'c_0\le1$ proves the $(1-\varepsilon)$ value guarantee.

It remains only to bound the menu size.  Since the chosen
$\tau^- =\rho_{j+1}\ge\rho_m$, it suffices that
\[
    k\ge\left\lceil\rho_m^{-1}\right\rceil.
\]
We show that the value of $k_\varepsilon$ in
\Cref{thm:status-product} meets this requirement.
Put $u_j:=\log(1/\rho_j)$.  From \eqref{eq:new-rho-rec} and
\[
    \Lambda_\eta(\rho)
    \le
    C\left(\rho^{-1}+\log\frac1\eta\right)
\]
we obtain
\[
    u_{j+1}
    \le
    B_\eta\left(u_j+C\log\frac C\eta\right).
\]
Thus
\[
    u_m
    \le
    C B_\eta^m\log\frac C\eta,
\]
and, since
\[
    B_\eta=O\!\left(\eta^{-2}\log^2\frac1\eta\right),
    \qquad
    m=O(\eta^{-1}),
\]
we have
\[
    \log u_m
    \le
    C\eta^{-1}\log\frac C\eta.
\]
Consequently, for a sufficiently large universal $C$,
\[
    \left\lceil\rho_m^{-1}\right\rceil
    \le
    \left\lceil
    \exp\!\left(
        \exp\!\left(
            C\varepsilon^{-1}\log\frac C\varepsilon
        \right)
    \right)
    \right\rceil
    =k_\varepsilon,
\]
as required.

\paragraph{Where the double exponential comes from.}
The locality/decorrelation argument costs a polynomial
$\Lambda_\eta(\tau)^{A_\eta}$ with
$A_\eta=O(\eta^{-2}\log^2(1/\eta))$.  Threshold selection then descends
through $m=\Theta(1/\eta)$ separated bands, each transition raising the previous scale to a power $B_\eta=A_\eta+O(1)$.  Hence
$\log(1/\rho_m)$ is exponential in
$\eta^{-1}\log(1/\eta)$, and converting back from the logarithmic scale
produces the double exponential.  Improving the universal rate
therefore seems to require either milder dependence of the light
completion on the crucial threshold or a way to avoid the
$\Theta(1/\eta)$-level threshold ladder.

\section{Local computation machinery}
\label{app:lca}

This appendix supplies the technical LCA arguments behind
\Cref{lem:prelim-fixed-correlation,thm:prelim-matching}.  We first state
an elementary composition rule and prove the fixed-input correlation
bound.  We then state the truncated random-greedy MIS primitive and
use it to construct the weighted matching rule.  Continuous edge
marginals are obtained by an internal perturbation of the weights.

\begin{lemma}[Composition of query profiles]
\label{lem:prelim-compose}
Let $\mathcal R$ be a natural local rule with query profile
$(P^+,P^-)$ and dictionary $\pi_0$.  Suppose an outer natural local
rule, using an independent fresh dictionary $\pi_1$, accesses
$\mathcal R$ only by issuing root queries.  Conditional on every
realization of $\pi_0$, assume that
\begin{itemize}[topsep=2pt,itemsep=1pt]
\item each output-root execution issues at most $B^+$ queries to
$\mathcal R$;
\item for every root $x$ of $\mathcal R$, the expected number, over
$\pi_1$, of queries to $\mathcal R$ at $x$ issued by all output-root
executions is at most $B^-$.
\end{itemize}
Then the composed local rule has query profile at most
\[
    (B^+P^+,\,B^-P^-).
\]
The same conclusion holds, up to adding the corresponding direct-probe
bounds, if the outer rule also probes the base graph directly.
\end{lemma}

\begin{proof}
For an output root $u$, let $C_{u,x}$ be the number of queries that its
execution issues to $\mathcal R$ at root $x$.  Pointwise,
\[
    |\mathrm{Read}_{\rm comp}(u)|
    \le
    \sum_x C_{u,x}\,
    |\mathrm{Read}_{\mathcal R}(x)|
    \le
    B^+P^+.
\]

For the reverse direction, fix a base vertex $w$ and condition on
$\pi_0$.  Then the read sets of $\mathcal R$ are fixed, and
\[
\begin{aligned}
    \mathbb E_{\pi_1}\!\left[
        \bigl|\{u:w\in\mathrm{Read}_{\rm comp}(u)\}\bigr|
    \right]
    &\le
    \sum_x
    \mathbb E_{\pi_1}\!\left[\sum_uC_{u,x}\right]
    \mathbf1\{w\in\mathrm{Read}_{\mathcal R}(x)\}
    \\
    &\le
    B^-
    \sum_x
    \mathbf1\{w\in\mathrm{Read}_{\mathcal R}(x)\}.
\end{aligned}
\]
Averaging over $\pi_0$ gives the bound $B^-P^-$.
\end{proof}

\fixedinputcorrelation*

\begin{proof}[Proof of \Cref{lem:prelim-fixed-correlation}]
Group all independent dictionary keys owned by one base vertex into
that vertex's private tape.  The hypotheses then match the private-tape
LCA model of \citet{ABGR25}.  Their Lemma~3.8 bounds the expected size
of the correlated set of any root by the product of the expected
out- and in-query bounds.  Since our out-query bound is pointwise,
this gives \eqref{eq:prelim-delta-sum}.  Their Lemma~3.7, specialized
to two roots, couples their two outputs to independent variables with
total-variation error at most a universal constant times
$\delta(u,u')$; for Bernoulli outputs this implies
\eqref{eq:prelim-cov-delta}.

Finally, expand the variance.  Using
$\Var(X_u)\le1$ and
$2|a_ua_{u'}|\le a_u^2+a_{u'}^2$, and then
\eqref{eq:prelim-delta-sum}, gives
\[
\begin{aligned}
    \Var_\pi\!\left(\sum_u a_uX_u\right)
    &\le
    \sum_u a_u^2
    +
    C\sum_{u\ne u'}|a_ua_{u'}|\delta(u,u')
    \\
    &\le
    C(1+R^+R^-)\sum_u a_u^2.
\end{aligned}
\]
Every query reads its own root, so $R^+,R^-\ge1$; absorbing the first
term proves \eqref{eq:prelim-fixed-variance}.
\end{proof}

\paragraph{Greedy independent set.}
Assign independent uniform priorities to the vertices, breaking ties
by a fixed ordering of their identifiers.  Write $w\prec_\pi u$ for
the resulting total order.  The \emph{random-greedy} independent set
$I=I(G,\pi)$ processes vertices in this order and inserts a vertex
exactly when none of its previously processed neighbors was inserted.
It is maximal, and
\[
    u\in I
    \iff
    \text{no neighbor $w\prec_\pi u$ of $u$ belongs to $I$}.
\]
The membership recursion visits earlier neighbors in increasing
priority order and stops at the first affirmative answer.  It follows
strictly decreasing paths and hence terminates, but its depth need
not be bounded independently of graph size.  The following primitive truncates
this recursion to obtain a deterministic out-query bound while deleting
only a small expected fraction of the exact random-greedy independent
set.

\begin{theorem}[Truncated random-greedy MIS
{\citep[Definition~6.7 and Claim~6.8]{ABGR25}}]
\label{thm:prelim-mis}
For every $\widehat\Delta\ge1$ and $\zeta\in(0,1)$ there is a natural
local rule on graphs of maximum degree at most $\widehat\Delta$,
computing an independent set
$I_\zeta=I_\zeta(G,\pi)$ such that
\begin{itemize}[topsep=2pt,itemsep=1pt]
\item[(i)]
for every dictionary realization $\pi$,
\[
    I_\zeta(G,\pi)\subseteq I(G,\pi),
\]
where $I(G,\pi)$ is the exact random-greedy maximal independent set
using the same priorities;

\item[(ii)]
for every $G$ of maximum degree at most $\widehat\Delta$,
\[
    \mathbb E_\pi[|I(G,\pi)\setminus I_\zeta(G,\pi)|]
    \le
    \zeta\,\mathbb E_\pi[|I(G,\pi)|];
\]

\item[(iii)]
the out-query complexity is
\[
    R^+=O(\widehat\Delta^3/\zeta);
\]

\item[(iv)]
the expected in-query complexity is
\[
    R^-=O(\widehat\Delta^2).
\]
\end{itemize}
In particular, $I_\zeta(G,\pi)$ is an independent set for every
$G$ and every $\pi$.
\end{theorem}

\begin{proof}
Simulate the exact membership recursion from the query root, aborting
the entire query and returning \emph{false} if it exceeds
$C_{\rm mis}\widehat\Delta^2/\zeta$ recursive calls.  Otherwise return
its exact answer.  Here $C_{\rm mis}$ is a sufficiently large universal
constant.  This is the truncation in
\citet[Definition~6.7]{ABGR25}.  Every affirmative answer is correct,
so $I_\zeta\subseteq I$ for every dictionary realization, proving (i).

By \citet[Lemma~6.6]{ABGR25}, querying the exact recursion once from
every root makes $O(n\widehat\Delta)$ recursive calls in expectation.
Consequently, choosing $C_{\rm mis}$ sufficiently large gives
\[
    \mathbb E_\pi[|\{u:\text{the query at $u$ truncates}\}|]
    \le \frac{\zeta n}{\widehat\Delta+1}
    \le \zeta\,\mathbb E_\pi[|I|],
\]
where the last inequality uses maximality of $I$.
Every vertex of $I\setminus I_\zeta$ is a truncating root, proving (ii).

Each recursive call requires at most $O(\widehat\Delta)$ vertex
probes to read the current vertex and its neighbors' priorities.
The call budget therefore gives (iii).
For (iv), distinguish these probes from recursive calls.
Lemma~6.6 also bounds the expected total number of calls to any fixed
vertex by $O(\widehat\Delta)$, over one execution from every root.
A vertex $w$ is probed only during a call to $w$ or one of its
neighbors.  Summing the call bound over these at most
$\widehat\Delta+1$ vertices gives $O(\widehat\Delta^2)$.
Truncation can only remove calls and probes.
\end{proof}

\paragraph{Bounded augmentations.}
For a matching $M$, an \emph{augmentation} is an $M$-alternating
path or cycle $P$ such that $M\triangle P$ is again a matching.  Its
gain is
\[
    g_M(P)
    :=
    w(P\setminus M)-w(P\cap M).
\]
A $k_0$-augmentation contains at most $k_0$ nonmatching edges.  It
therefore contains $O(k_0)$ vertices.

\begin{lemma}[Comparison lemma {\citep[Theorem~2.1]{PS04}}]
\label{lem:ps04}
Let $M$ be a matching of a weighted graph and $k_0\ge1$.  There is a
vertex-disjoint family $\mathcal A$ of $k_0$-augmentations such that
\[
    g_M(\mathcal A)
    \ge
    \frac{k_0+1}{2k_0+1}
    \left(
        \frac{k_0}{k_0+1}\operatorname{OPT}_w-w(M)
    \right).
\]
\end{lemma}

\weightedmatchingrule*

\begin{proof}
We first analyze the construction with a working accuracy
$\alpha\in(0,1/4)$ and a lower weight bound $b$ satisfying
\[
    \alpha\le b\le1.
\]
This gives a base rule $\Psi^{\mathrm{base}}_{\alpha,b}$.
We will then choose the internal accuracy and perturb the input
weights to obtain a single rule with both properties of the theorem.
Throughout the analysis of the base rule, $G$ and $w$ are fixed,
$w_e\in[b,1]$, and all randomness belongs to the priority
dictionary $\pi_0$.

Write
\[
    W:=\operatorname{OPT}_w(G).
\]
If $G$ has no edges, the empty matching satisfies all claims.
Henceforth assume $W>0$.

\paragraph{The construction.}
Fix
\[
    k_0=\left\lceil\frac{C_0}{\alpha}\right\rceil,
    \qquad
    \zeta=c_{\mathrm{mis}}\alpha\le\frac14,
\]
where $C_0$ is a sufficiently large universal constant and
$c_{\mathrm{mis}}>0$ is a sufficiently small universal constant.
Recall that a $k_0$-augmentation has at most $k_0$
nonmatching edges and therefore has $O(k_0)$ vertices.
Its gain is at most $k_0$.

The algorithm starts from the empty matching and performs
a fixed number of sweeps.  Each sweep processes the dyadic
gain buckets
\[
    [2^j\alpha b,\,2^{j+1}\alpha b),
    \qquad j=0,\ldots,J-1,
\]
in decreasing order of $j$, where
\[
    J
    :=1+\left\lfloor
        \log_2\frac{k_0}{\alpha b}
    \right\rfloor
    =O\!\left(\log\frac{e}{\alpha}\right).
\]
The last bound uses $b\ge\alpha$.

At a bucket stage, consider all currently valid
$k_0$-augmentations whose gains lie in that bucket.
Form their conflict graph, joining two candidates when
they share a base-graph vertex.
Using fresh independent priorities, let $I$ be the exact
random-greedy maximal independent set of this graph,
and let $I_\zeta\subseteq I$ be the truncated set supplied
by \Cref{thm:prelim-mis}.
Apply all augmentations in $I_\zeta$ simultaneously.
They are vertex-disjoint, so the result is a matching.
All applied augmentations have positive gain.

\paragraph{Progress in one sweep.}
Let $M_s$ be the matching at the start of sweep $s$, and define
\[
    \mathrm{gap}(M)
    :=\frac{k_0}{k_0+1}W-w(M).
\]
By \Cref{lem:ps04}, when $\mathrm{gap}(M_s)>0$ there is a
vertex-disjoint family $\mathcal A^\star$ of
$k_0$-augmentations with
\[
    g_{M_s}(\mathcal A^\star)
    \ge\frac12\,\mathrm{gap}(M_s).
\]
When $\mathrm{gap}(M_s)\le0$, the same inequality holds
with $\mathcal A^\star=\varnothing$.
Discarding members of nonpositive gain can only increase
the total gain, so assume every member has positive gain.

Each member then contains a nonmatching edge.
Choosing one such edge from each member gives a matching,
and therefore
\[
    |\mathcal A^\star|
    \le\nu(G)
    \le\frac{W}{b}.
\]
Discard from this comparison family all augmentations
of gain below $\alpha b$.
Their total gain is at most
\[
    \alpha b\,|\mathcal A^\star|
    \le\alpha W.
\]
Thus the surviving comparison family has total gain at least
\[
    \frac12\,\mathrm{gap}(M_s)-\alpha W.
\]
This family is used only in the analysis; the algorithm
considers all valid augmentations in each bucket.

Fix all randomness in the sweep.
For a surviving comparison augmentation $A$, if an applied
augmentation from an earlier bucket meets $A$, charge $A$
to the first such augmentation.
Otherwise no previous update has touched a vertex of $A$.
Consequently, when the bucket containing $g_{M_s}(A)$ is
processed, $A$ is still valid and has the same gain.
By maximality of the exact greedy set $I$ in that bucket,
$A$ either belongs to $I$ or meets a member of $I$.
Charge $A$ to such a member.

In either case, the augmentation receiving the charge has
gain at least $g_{M_s}(A)/2$.
Moreover, each candidate contains $O(k_0)$ vertices,
whereas the comparison augmentations are vertex-disjoint.
Hence each candidate receives at most $O(k_0)$ charges.

Let
\[
    G_s:=w(M_{s+1})-w(M_s)
\]
be the actual gain of the sweep.
Let $D_s$ be the sum, over its bucket stages, of the gains
of candidates in $I\setminus I_\zeta$.
The charging argument gives the pointwise bound
\[
    \frac12\,\mathrm{gap}(M_s)-\alpha W
    \le Ck_0(G_s+D_s).
\]
Equivalently, for universal constants $c,C>0$,
\begin{equation}
\label{eq:prelim-sweep}
    G_s
    \ge
    \frac{c}{k_0}\,\mathrm{gap}(M_s)
    -\frac{C\alpha}{k_0}W-D_s.
\end{equation}

\paragraph{Loss from truncation.}
Consider a single bucket stage, conditional on its history
before the fresh priority dictionary is drawn.
The candidate graph and all candidate gains are then fixed.
Write
\[
    G_j^\star:=\sum_{A\in I}g(A),
    \qquad
    G_j:=\sum_{A\in I_\zeta}g(A),
    \qquad
    D_j:=\sum_{A\in I\setminus I_\zeta}g(A).
\]
Pointwise, $G_j^\star=G_j+D_j$.
Since gains within a bucket differ by at most a factor two,
\Cref{thm:prelim-mis}(ii) implies
\[
    \mathbb E_{\pi_0}[D_j\mid\mathrm{past}]
    \le
    2\zeta\,
    \mathbb E_{\pi_0}[G_j^\star\mid\mathrm{past}].
\]
Using $\zeta\le1/4$ and rearranging gives
\[
    \mathbb E_{\pi_0}[D_j\mid\mathrm{past}]
    \le
    4\zeta\,
    \mathbb E_{\pi_0}[G_j\mid\mathrm{past}].
\]

Let $D_{\mathrm{tot}}:=\sum_sD_s$ and $G_{\mathrm{tot}}:=\sum_sG_s$.
Summing over all stages and using the tower property yields
\begin{equation}
\label{eq:prelim-total-trunc}
    \mathbb E_{\pi_0}[D_{\mathrm{tot}}]
    \le
    4\zeta\,\mathbb E_{\pi_0}[G_{\mathrm{tot}}]
    \le
    4\zeta W
    =O(\alpha)W.
\end{equation}
Here $G_{\mathrm{tot}}\le W$ pointwise because the algorithm
starts from the empty matching and applies only
positive-gain augmentations.

\paragraph{Number of sweeps and approximation guarantee.}
Since
\[
    \mathrm{gap}(M_{s+1})
    =\mathrm{gap}(M_s)-G_s,
\]
\eqref{eq:prelim-sweep} implies
\[
    \mathrm{gap}(M_{s+1})
    \le
    \left(1-\frac{c}{k_0}\right)\mathrm{gap}(M_s)
    +\frac{C\alpha}{k_0}W+D_s.
\]
Choose
\[
    L
    =\left\lceil
        C_L k_0\log\frac{e}{\alpha}
      \right\rceil
\]
sweeps, with $C_L$ sufficiently large.
Unrolling the recurrence, the initial term contributes
$O(\alpha)W$, and the geometric sum of the deterministic
additive terms contributes $O(\alpha)W$.
The expected contribution of the $D_s$ terms is at most
$\mathbb E_{\pi_0}[D_{\mathrm{tot}}]=O(\alpha)W$ by
\eqref{eq:prelim-total-trunc}.
Thus
\[
    \mathbb E_{\pi_0}[\mathrm{gap}(M_L)]
    \le O(\alpha)W.
\]
Since $k_0/(k_0+1)=1-O(\alpha)$, there is a universal
constant $K\ge1$ such that
\[
    \mathbb E_{\pi_0}[w(M_L)]
    \ge(1-K\alpha)W.
\]
The total number of sequential bucket stages is
\begin{equation}
\label{eq:prelim-stages}
    S=LJ
    =O\!\left(
        \alpha^{-1}\log^2\frac{e}{\alpha}
      \right).
\end{equation}

\paragraph{Local implementation.}
For fixed $k_0$, the number of potential augmentations
containing any given base-graph vertex is
$\Delta^{O(k_0)}$.
The same bound holds for the number of potential
augmentations intersecting a given augmentation.
Consequently every candidate conflict graph has maximum
degree at most
\[
    \widehat\Delta=\Delta^{O(k_0)}.
\]

Priority marks are keyed by the stage and a canonical
encoding of the augmentation, and are owned by a canonical
vertex of its support.
The key universe can be fixed in advance using stage-tagged
vertex sequences of length $O(k_0)$.
Invalid encodings are never used as candidates.
All priority marks are independent, and their collection
is the shared randomness $\pi_0$.

Represent each intermediate matching by a local rule that,
when queried at a vertex, returns its matched partner
or declares it free.
To answer an updated query at a root $u$, enumerate the
potential augmentations containing $u$, determine which
are valid candidates, and query their membership in
$I_\zeta$.
Their validity and gains are determined by local edge data
and queries to the preceding matching rule at their
support vertices.
Since selected augmentations are vertex-disjoint,
these answers, together with the preceding matching
at $u$, determine the updated partner of $u$.

A probe of the candidate conflict graph is simulated by
enumerating intersecting potential augmentations and
checking their validity.
All these enumerations and local validations have
multiplicity $\Delta^{O(k_0)}$.
Within each base-root execution, cache simulated candidate probes and
answers from the preceding matching rule, so that repeated requests
do not repeat their underlying probes.  The cache is private to that
execution and does not change the MIS recursion or its call budget.
In particular, a fixed base location can be inspected
by the direct local validation of only
$\Delta^{O(k_0)}$ potential candidates.

\paragraph{Query-profile accounting.}
Suppose the preceding matching rule has query profile
$(P^+,P^-)$.
We apply \Cref{lem:prelim-compose} to one bucket stage.

For the out-query direction, a base-root execution launches
at most $\Delta^{O(k_0)}$ truncated-MIS root queries.
By \Cref{thm:prelim-mis}(iii), each makes at most
$O(\widehat\Delta^3/\zeta)$ conflict-graph probes.
Each such probe requires at most $\Delta^{O(k_0)}$
queries to the preceding matching rule.
Including the direct candidate validations gives
\[
    B^+\le\frac{\Delta^{O(k_0)}}{\zeta}.
\]

For the reverse direction, condition on all preceding-stage
dictionaries.
The current candidate graph and the lists of candidate
roots queried from each base root are now fixed.
Only the fresh priorities of the current stage remain random.

Each candidate appears as a truncated-MIS query root
for at most $\Delta^{O(k_0)}$ base output roots.
If every candidate root were queried once,
\Cref{thm:prelim-mis}(iv) would bound the expected number
of executions probing any fixed candidate by
$O(\widehat\Delta^2)$.
Accounting for the bounded repetition of candidate roots
therefore gives at most
\[
    \Delta^{O(k_0)}\,O(\widehat\Delta^2)
\]
expected probes of any fixed candidate.

Fix a root $x$ of the preceding matching rule.
A query at $x$ can arise either from a direct candidate
validation at a base output root, or while simulating a
conflict-graph probe.
The first source has multiplicity $\Delta^{O(k_0)}$.
For the second source, only $\Delta^{O(k_0)}$ conflict
candidates have local probe simulations that can reach $x$,
and each such simulation issues at most
$\Delta^{O(k_0)}$ queries to the preceding matching rule.
Thus, over one execution from every base output root,
the expected number of preceding-rule queries at $x$ is
\[
    B^-\le\Delta^{O(k_0)}.
\]
This bound holds conditional on every preceding-stage
dictionary, as required by \Cref{lem:prelim-compose}.

Direct base-graph probes, including probes needed to read
locally owned marks, satisfy the same bounds.
By the additive clause of \Cref{lem:prelim-compose},
each stage therefore multiplies both profile coordinates,
after adding one to each, by at most
\[
    B:=\frac{\Delta^{O(k_0)}}{\zeta}.
\]
Since $\Delta\ge2$, $k_0=\Theta(1/\alpha)$, and
$\zeta^{-1}=O(1/\alpha)$, we can write
$B=\Delta^{O(k_0)}$.
Starting from the empty matching and composing the
$S$ stages gives
\[
    R^+,R^-
    \le
    \Delta^{O(k_0S)}
    =
    \Delta^{O\!\left(
        \alpha^{-2}\log^2\frac{e}{\alpha}
    \right)}.
\]

The construction is natural.
Candidate enumeration follows bounded paths from the
query root, conflict-graph steps move between intersecting
augmentations, and every preceding-rule query is rooted
at a vertex exposed by this connected exploration.
A priority mark is read only after its owner has been probed.

\paragraph{The final rule: internal perturbations and accuracy.}
Let $K\ge1$ be the universal constant in the base approximation
guarantee above.  Fix
\[
    \alpha:=\frac{\eta}{2(K+1)},
    \qquad
    \varpi:=\frac{\eta}{2},
\]
and take $c_{\mathrm{lca}}\ge2$.
For every edge $e$, draw an independent mark
\[
    \rho_e\sim\operatorname{Uniform}[0,\varpi],
\]
independently of the priority dictionary $\pi_0$.
Each mark is owned by a fixed endpoint of its edge and is read
only after that owner is probed.
Define
\[
    \widetilde w_e
    :=\frac{w_e(1+\rho_e)}{1+\varpi},
    \qquad
    \widetilde b:=\frac{\gamma}{1+\varpi}.
\]
Then $\widetilde w_e\in[\widetilde b,1]$, and
\[
    \varpi\le\frac\eta2<\frac14,
    \qquad
    \widetilde b\ge\frac45\gamma\ge\alpha.
\]
Thus the base construction applies at working accuracy $\alpha$
and lower weight bound $\widetilde b$.
With the full dictionary $\pi:=(\pi_0,\rho)$, define the final rule by
\[
    M=\Psi(G,\pi;w)
    :=\Psi^{\mathrm{base}}_{\alpha,\widetilde b}
        (G,\pi_0;\widetilde w).
\]
The parameters $\alpha,\varpi,\widetilde b$ depend only on
$\eta,\gamma$, and all marks have fixed local owners.

Conditional on $\rho$, the perturbed weights are fixed and
$\pi_0$ retains its original distribution.
The base query-profile bounds therefore hold uniformly in $\rho$.
Reading the perturbation marks adds only bounded local probe overhead.
Averaging over $\rho$ preserves the expected in-query bound, while
the out-query bound holds for every full dictionary realization.
Since $\alpha$ is a fixed constant multiple of $\eta$, a sufficiently
large universal $C_{\mathrm{lca}}$ gives
\[
    R^+,R^-\le\Delta^{a_\eta},
    \qquad
    a_\eta=C_{\mathrm{lca}}\eta^{-2}\log^2\frac{e}{\eta}.
\]
This also accounts for the perturbation marks within the local-access
model of \Cref{def:prelim-access}.

Conditional on $\rho$, the base approximation guarantee gives
\[
    \mathbb E_{\pi_0}[\widetilde w(M)\mid\rho]
    \ge(1-K\alpha)\operatorname{OPT}_{\widetilde w}(G)
    \ge(1-\eta/2)\operatorname{OPT}_{\widetilde w}(G).
\]
Since $w_e\ge\widetilde w_e\ge w_e/(1+\varpi)$,
\[
\begin{aligned}
    \mathbb E_\pi[w(M)]
    &\ge
    \mathbb E_\rho\!\left[
        \mathbb E_{\pi_0}[\widetilde w(M)\mid\rho]
    \right]\\
    &\ge
    (1-\eta/2)\,
    \mathbb E_\rho[\operatorname{OPT}_{\widetilde w}(G)]\\
    &\ge
    \frac{1-\eta/2}{1+\varpi}\operatorname{OPT}_w(G)\\
    &\ge
    (1-\eta)\operatorname{OPT}_w(G),
\end{aligned}
\]
where the last inequality uses $\varpi\le\eta/2$.
Together with the query-profile bounds, this proves part~(i).

\paragraph{Continuity of the output marginals.}
Fix $G,\eta,\gamma$ and a weight vector
\[
    w^0\in[\gamma,1]^{E(G)}.
\]
In particular, the working accuracy, perturbation range, lower
weight bound $\widetilde b$, number of stages, and bucket boundaries
are fixed independently of the varying weight vector $w$.

There are finitely many possible intermediate matchings,
bounded augmentations, and gain cutoffs.
For any fixed matching $N$ and nonempty augmentation $P$,
its perturbed gain at $w^0$ is
\[
    \widetilde g_N(P;w^0,\rho)
    =
    \frac{1}{1+\varpi}
    \left(
        \sum_{e\in P\setminus N}w^0_e(1+\rho_e)
        -
        \sum_{e\in P\cap N}w^0_e(1+\rho_e)
    \right).
\]
Because $w^0_e>0$, this is a nonconstant affine function
of the independent continuous marks $\rho_e$.
It therefore equals any fixed gain cutoff or bucket
boundary with probability zero.

Take a finite union over all possible intermediate
matchings, augmentations, and relevant boundaries.
With probability one, none of these gain comparisons
is tight at $w^0$.
Priority ties also have probability zero.
This argument considers all possible discrete histories
at once and does not condition on the realized history.

Fix a full dictionary realization $\pi=(\pi_0,\rho)$ outside
this null event.  All relevant gain comparisons are strict at $w^0$
and depend continuously on $w$.
Inducting over the finitely many stages, the candidate
graphs, truncated-MIS decisions, and resulting matchings
remain unchanged throughout some neighborhood of $w^0$.
In particular, for every $e\in E(G)$,
\[
    \mathbf1\{e\in\Psi(G,\pi;w)\}
\]
is locally constant at $w^0$ for almost every $\pi$.

The indicator is bounded by one.
Bounded convergence therefore shows that
\[
    w\longmapsto\Pr_\pi[e\in\Psi(G,\pi;w)]
\]
is continuous at $w^0$.
Since $w^0$ was arbitrary, this proves part~(ii).
\end{proof}

\section{Decorrelation under product arrivals}
\label{app:decorrelation}

This appendix proves the transfer theorem used in the light completion.
The only additional ingredient beyond the fixed-input correlation bound
is the stability of the capped graph under resampling one arrival state.

\paragraph{State-resampling stability.}
For the bounds below, which average over arrivals and auxiliary marks,
we separate these two sources of randomness.  Conditional on the arrival
states $T$, the capped graph is
fixed and only the locally owned dictionary of the witness remains
random.  To control the remaining variation in $T$, we resample one
arrival state at a time while keeping the entire witness dictionary
fixed.

\begin{corollary}[State-resampling stability]
\label{cor:prelim-blockstab}
Let $T'$ be obtained from $T$ by replacing a single state $T_i$ by an
independent copy $T'_i$, while keeping the witness dictionary fixed.
Then
\[
    \bigl|
        E(\mathcal G_C^D(T))
        \triangle
        E(\mathcal G_C^D(T'))
    \bigr|
    \le
    2dD.
\]
Let $X_i(T,T')$ be the set of vertices whose capped neighborhood
changes, as in \eqref{eq:cap-vertex-stability}, and put
\[
    \widehat X_i(T,T'):=X_i(T,T')\cup\{i\}.
\]
With the state-edge storage convention of
\Cref{subsec:witness-construction},
\begin{equation}
\label{eq:transfer-X-size}
    |\widehat X_i(T,T')|
    \le
    K:=1+2d(D+1),
\end{equation}
all capped adjacencies and all state-dependent local data are identical
outside $\widehat X_i(T,T')$, and $\widehat X_i(T,T')$ lies within
distance two of $i$ in the union of the two realized crucial graphs.
Consequently, under the coupling that uses the same dictionary in the
two environments, the answer of a natural query can change only if its
execution reads a vertex of $\widehat X_i(T,T')$.
\end{corollary}

\begin{proof}
The edge bound and the neighborhood-change set are given by
\Cref{lem:capping}(ii).  The proof of that lemma bounds the containing
set $\{i\}\cup X_i(T,T')$ by $1+2d(D+1)$, so adjoining $i$ does not
increase the displayed bound.  State-dependent state-edge data are
stored at their arrival endpoint, while all auxiliary marks are held
fixed.  Hence outside $\widehat X_i$ the two executions see identical
adjacency, local data, and dictionary values.  A first divergence can
therefore occur only when a vertex of $\widehat X_i$ is read.
\end{proof}

\decorrelationtransfer*

\begin{proof}
Let $C_{\mathrm{corr}}\ge1$ be a universal constant for the
variance bound in \Cref{lem:prelim-fixed-correlation}.
We will show that $C_{\mathrm{dec}}:=C_{\mathrm{corr}}+14$ suffices.
We separate the witness-dictionary randomness from the arrival-state
randomness.  Conditional on a realization of $T$, the capped graph and
all state-dependent local data are fixed, and the witness is a
private-tape local rule in the sense of
\Cref{lem:prelim-fixed-correlation}.  Hence, for every deterministic
coefficient vector $(a_u)$,
\begin{equation}
\label{eq:transfer-fixed-input}
    \Var_\pi\!\left(
        \sum_u a_uF_u
        \,\middle|\,
        T
    \right)
    \le
    C_{\mathrm{corr}} R^+R^-\sum_u a_u^2.
\end{equation}

For the transfer over $T$, fix a realization $T$ and put
\[
    L^T_{u,w}
    :=
    \Pr_\pi[w\in\mathrm{Read}_T(u)].
\]
By the query-profile bounds,
\begin{equation}
\label{eq:transfer-L-sums}
    \sum_wL^T_{u,w}\le R^+,
    \qquad
    \sum_uL^T_{u,w}\le R^-.
\end{equation}
Indeed, the first inequality follows from the pointwise out-query
bound, and the second is precisely the expected in-query bound on the
fixed input $T$.  Schur's test therefore gives, for every deterministic
$(a_u)$,
\begin{equation}
\label{eq:transfer-schur}
    \sum_w
    \left(
        \sum_u |a_u|L^T_{u,w}
    \right)^2
    \le
    R^+R^-\sum_u a_u^2.
\end{equation}

We next record the multiplicity with which one fixed input location can
be affected by a single-coordinate resampling.  Let $T^{(j)}$ be
obtained by replacing $T_j$ by an independent draw $T'_j\sim p_j$, and
let
\[
    \widehat X_j
    :=
    \widehat X_j(T,T^{(j)})
\]
be the state-resampling set of
\Cref{cor:prelim-blockstab}.  Besides
$|\widehat X_j|\le K:=1+2d(D+1)$, we have, for every fixed $T$ and
input vertex $w$,
\begin{equation}
\label{eq:transfer-influence-multiplicity}
    \sum_j
    \Pr[\,w\in\widehat X_j\mid T\,]
    \le
    M:=1+d(D+d+1).
\end{equation}
For an offline vertex $v$, a single-coordinate resampling cannot change
its capped neighborhood when $N_v(T)\ge D+2$.  Otherwise there are at
most $D+1$ arrivals currently realizing a crucial edge at $v$, whereas
\[
    \sum_j
    \Pr[(j,T'_j,v)\in C_\tau]
    =
    \sum_{j,t}p_j(t)
      \mathbf1\{z_{jtv}\ge\tau\}
    \le
    \frac1\tau
    \sum_{j,t}p_j(t)z_{jtv}
      \mathbf1\{z_{jtv}\ge\tau\}
    =
    \frac{c_v}{\tau}
    \le d.
\]
Thus the left side of
\eqref{eq:transfer-influence-multiplicity} is at most $D+d+1$ when
$w=v$ is offline.  Now let $w=r$ be an arrival.  The resampling
$j=r$ contributes at most one.  For $j\ne r$, the capped neighborhood
of $r$ can change only through one of the at most $d$ crucial offline
neighbors of its fixed state $T_r$; applying the preceding offline
bound to each such neighbor gives at most
$1+d(D+d+1)=M$.  This proves
\eqref{eq:transfer-influence-multiplicity}.

We first prove the offline variance bound.  Put
\[
    S(T,\pi):=\sum_v h_vF_v,
    \qquad
    f(T):=\mathbb E_\pi[S(T,\pi)].
\]
Couple $T$ and $T^{(j)}$ using the same witness dictionary.  By
\Cref{cor:prelim-blockstab}, if the two answers at an offline root $v$
differ then, before their first divergence, the execution on input $T$
must read a vertex of $\widehat X_j$.  Hence
\[
    \Pr_\pi[
        F_v(T)\ne F_v(T^{(j)})
    ]
    \le
    \sum_{w\in\widehat X_j}L^T_{v,w},
\]
and therefore
\begin{equation}
\label{eq:transfer-first-divergence}
    |f(T)-f(T^{(j)})|
    \le
    \sum_{w\in\widehat X_j}
    \sum_v |h_v|L^T_{v,w}.
\end{equation}
For nonnegative numbers $b_w$ and a set $X$ of size at most $K$,
$(\sum_{w\in X}b_w)^2\le K\sum_{w\in X}b_w^2$.  Squaring
\eqref{eq:transfer-first-divergence}, summing over $j$, and taking the
conditional expectation over the resampled states gives
\[
\begin{aligned}
    &\sum_j
    \mathbb E[
        (f(T)-f(T^{(j)}))^2
        \mid T
    ]
    \\
    &\qquad\le
    K\sum_w
    \left(\sum_v|h_v|L^T_{v,w}\right)^2
    \sum_j\Pr[w\in\widehat X_j\mid T]
    \\
    &\qquad\le
    KM R^+R^-\sum_vh_v^2,
\end{aligned}
\]
where the last line uses
\eqref{eq:transfer-influence-multiplicity} and
\eqref{eq:transfer-schur}.  Efron--Stein now yields
\[
    \Var_T f(T)
    \le
    \frac{KM}{2}R^+R^-\sum_vh_v^2.
\]
Combining this with the conditional variance bound
\eqref{eq:transfer-fixed-input} through the law of total variance gives
\[
    \Var\!\left(\sum_vh_vF_v\right)
    \le
    (C_{\mathrm{corr}}+KM/2)R^+R^-\sum_vh_v^2.
\]
Since
\[
    K M
    \le
    10d^2D(D+d+1),
\]
our choice of $C_{\mathrm{dec}}$ proves the offline variance bound in \Cref{thm:prelim-transfer}.  If one state $T_i=t$ is fixed, the same
argument applies to the remaining independent coordinates and simply
omits $i$ from Efron--Stein, so the same bound holds conditionally on
$T_i=t$.

For the online form, define
\[
    S_{\rm on}(T,\pi)
    :=
    \sum_i h_{i,T_i}F_i,
    \qquad
    f_{\rm on}(T)
    :=
    \mathbb E_\pi[S_{\rm on}(T,\pi)].
\]
Conditional on $T$, \eqref{eq:transfer-fixed-input} with
$a_i=h_{i,T_i}$ gives
\[
    \Var_\pi(S_{\rm on}\mid T)
    \le
    C_{\mathrm{corr}} R^+R^-\sum_i h_{i,T_i}^2.
\]
Under the coupling of $T$ and $T^{(j)}$, decompose
\[
\begin{aligned}
    |f_{\rm on}(T)-f_{\rm on}(T^{(j)})|
    &\le
    \sum_i |h_{i,T_i}|
       \Pr_\pi[F_i(T)\ne F_i(T^{(j)})]
    \\
    &\qquad+
    |h_{j,T_j}-h_{j,T'_j}|.
\end{aligned}
\]
The square of the first term is handled exactly as in
\eqref{eq:transfer-first-divergence}, now with coefficients
$|h_{i,T_i}|$.  The sum of its expected squares over $j$ is at most
\[
    KM R^+R^-
    \sum_{i,t}p_i(t)h_{it}^2.
\]
For the direct coefficient change,
\[
    \sum_j
    \mathbb E[|h_{j,T_j}-h_{j,T'_j}|^2]
    \le
    2\sum_j
      \mathbb E[h_{j,T_j}^2+h_{j,T'_j}^2]
    =
    4\sum_{j,t}p_j(t)h_{jt}^2.
\]
Using $(a+b)^2\le2a^2+2b^2$, Efron--Stein, and the law of total
variance gives
\[
    \Var(S_{\rm on})
    \le
    (C_{\mathrm{corr}}+KM+4)R^+R^-
    \sum_{i,t}p_i(t)h_{it}^2,
\]
where $R^+R^-\ge1$ absorbs the direct coefficient-change term.
The bounds on $KM$ and the choice of $C_{\mathrm{dec}}$ imply
the claimed online-state variance bound with the same $\Gamma$.

It remains to prove the aggregate conditioning-bias bound.
Fix $(i,t)$ and couple an unconditional profile $T$ to
$T^{i\leftarrow t}$, obtained by replacing only $T_i$ by $t$, using
the same dictionary.  Put
\[
    A_{it}:=\{c_{i,T_i}>0\text{ or }c_{it}>0\}.
\]
Outside this event, $i$ is isolated in both crucial graphs and its
state change cannot affect any offline output of a natural rule.
The first-divergence argument and the column bound in
\eqref{eq:transfer-L-sums} therefore give
\begin{equation}
\label{eq:transfer-bias-one-state}
\begin{aligned}
    &\left|
        \sum_vh_{itv}
        \bigl(\Pr[F_v=1\mid T_i=t]-\Pr[F_v=1]\bigr)
      \right|\\
    &\qquad\le
    h_{\max}R^-
    \mathbb E\!\left[
        |\widehat X_i(T,T^{i\leftarrow t})|
        \mathbf1_{A_{it}}
    \right].
\end{aligned}
\end{equation}
Since $c_{it}>0$ implies $c_{it}\ge\tau$,
\[
    \sum_{i,t}p_i(t)\mathbf1\{c_{it}>0\}
    \le\frac{x(C_\tau)}{\tau}
    \le d\,x(C_\tau).
\]
Using $|\widehat X_i|\le K$ and averaging over both the old and
target states consequently gives
\[
    \sum_{i,t}p_i(t)
    \mathbb E\!\left[
        |\widehat X_i(T,T^{i\leftarrow t})|\mathbf1_{A_{it}}
    \right]
    \le2Kd\,x(C_\tau)
    \le10d^2D\,x(C_\tau).
\]
Summing \eqref{eq:transfer-bias-one-state} with weights $p_i(t)$ and
using $C_{\mathrm{dec}}\ge10$ proves the bias bound.
\end{proof}

\section{Hub instances: the universal \texorpdfstring{$\Omega(1/k)$}{Omega(1/k)} loss}
\label{sec:lower-bound}

For an instance $\mathcal I$, let $\mathcal G$ be its full realized graph and let
$H_\Phi$ be the graph retained by a menu rule $\Phi$.  When
$\mathbb E[\nu(\mathcal G)]>0$, define
\[
    \alpha_k(\mathcal I)
    :=
    \sup_{\Phi:\,\mathrm{menu}(\Phi)\le k}
    \frac{\mathbb E[\nu(H_\Phi)]}{\mathbb E[\nu(\mathcal G)]},
    \qquad
    \alpha^\star(k):=\inf_{\mathcal I}\alpha_k(\mathcal I),
\]
where the infimum is over nontrivial product-arrival instances.  Thus
$\alpha^\star(k)$ is the best preservation guarantee obtainable uniformly over all
such instances.  By \Cref{thm:status-product}, $\alpha^\star(k)\to1$; quantitatively,
for all sufficiently large $k$,
\begin{equation}
\label{eq:intro-upper-rate}
    1-\alpha^\star(k)
    \le
    O\!\left(\frac{\log\log\log k}{\log\log k}\right).
\end{equation}

We show that the universal convergence cannot be faster than order $1/k$.  This holds
for \emph{independent presence} and applies to all menu rules, not only
benchmark-derived ones. The construction uses the following class of instances.

\paragraph{Hub instances.}
We call $\mathcal I$ a \emph{hub instance} if its offline vertices $v_1,\dots,v_m$
(the \emph{columns}) each have a \emph{dedicated arrival} $u_j$, present with
probability $\sigma_j\in[0,1]$ and adjacent only to $v_j$, and its remaining
arrivals are $h\ge1$ \emph{hubs}, hub $i$ being present with probability
$\pi_i\in(0,1]$ and adjacent to every column.  All arrivals are independent.  A hub adds value only when it finds a free
column, one whose dedicated arrival is absent.

\begin{theorem}[Multi-hub lower bound]
\label{thm:needy-lower}
Let $k,h\ge1$ be integers, and let $\mathcal I$ be the hub
instance with $m=(k+1)h$ columns, $\sigma_j=k/(k+1)$ for
all $j$, and $h$ always-present hubs. Every menu rule of
size at most $k$ satisfies
\[
    \mathbb E[\nu(H)]
    \le
    \left(
        1-\frac{(k/(k+1))^k-h^{-1/2}}{k+1}
    \right)\mathbb E[\nu(\mathcal G)].
\]
Consequently, letting $h\to\infty$,
\begin{equation}
\label{eq:intro-lower-rate}
    1-\alpha^\star(k)
    \ge \frac{k^k}{(k+1)^{k+1}}
    \ge \frac{1}{e(k+1)},
    \qquad k\ge1.
\end{equation}
The first lower bound is asymptotic to $1/(ek)$.
\end{theorem}
\begin{proof}
We may assume that every dedicated edge is retained, since adding
these edges respects the menu budget and can only increase
$\nu(H)$. In either graph, there is a maximum matching that
matches every present dedicated arrival: if such an arrival is
unmatched, its column must be matched to a hub, and we can
replace that edge by the dedicated edge.

Let $N$ be the set of columns whose dedicated arrivals are absent.
Then
\[
    |N|\sim\mathrm{Bin}\!\left(m,\frac1{k+1}\right),
    \qquad
    \mathbb E[|N|]=h,
    \qquad
    \Var(|N|)\le h.
\]
Matching the present dedicated arrivals first gives
\[
    \nu(\mathcal G)=m-|N|+\min\{h,|N|\},
    \qquad
    \mathbb E[\nu(\mathcal G)]\le m=(k+1)h.
\]

Let $D_i$ be hub $i$'s menu. In a maximum matching of $H$
that matches all present dedicated arrivals, each matched hub
uses a column of $N$. Consequently,
\[
    \nu(H)
    \le m-|N|+\sum_{i=1}^h
        \mathbf1\{D_i\cap N\ne\varnothing\}.
\]
Each hub has a single state, so its menu is independent of $N$.
Conditional on $D_i$, independence of the dedicated arrivals
and $|D_i|\le k$ give
\[
    \Pr[D_i\cap N=\varnothing\mid D_i]
    =
    \left(\frac{k}{k+1}\right)^{|D_i|}
    \ge
    \left(\frac{k}{k+1}\right)^k.
\]
It follows that
\[
\begin{aligned}
    \mathbb E[\nu(\mathcal G)-\nu(H)]
    &\ge
    \mathbb E[\min\{h,|N|\}]
    -h\left(1-\left(\frac{k}{k+1}\right)^k\right)\\
    &=
    h\left(\frac{k}{k+1}\right)^k
    -\mathbb E[(h-|N|)_+]         \ge
    h\left(\frac{k}{k+1}\right)^k-\sqrt h.
\end{aligned}
\]
The last inequality uses
\[
    \mathbb E[(h-|N|)_+]
    \le \mathbb E\bigl[\,\bigl||N|-h\bigr|\,\bigr]
    \le \sqrt{\Var(|N|)}
    \le \sqrt h.
\]
If the resulting lower bound on the loss is nonpositive, the
claimed inequality follows from $\nu(H)\le\nu(\mathcal G)$.
Otherwise, dividing by
$\mathbb E[\nu(\mathcal G)]\le(k+1)h$ gives
\[
    \frac{\mathbb E[\nu(\mathcal G)-\nu(H)]}
         {\mathbb E[\nu(\mathcal G)]}
    \ge
    \frac{(k/(k+1))^k-h^{-1/2}}{k+1},
\]
which proves the first claim.

Letting $h\to\infty$ and using the definition of
$\alpha^\star(k)$ yields
\[
    1-\alpha^\star(k)
    \ge
    \frac{k^k}{(k+1)^{k+1}}
    =
    \frac{1}{(k+1)(1+1/k)^k}
    \ge
    \frac{1}{e(k+1)}.
\]
Here $(1+1/k)^k\le e$ for every $k\ge1$, and its convergence
to $e$ shows that the first lower bound is asymptotic to
$1/(ek)$.
\end{proof}

The bound stems from a coverage effect: Any menu of $k$ columns is fixed before the
arrivals are seen, while each column turns out to be free with probability only
$1/k$; so a hub's menu contains no free column with probability $(1-1/k)^k\to1/e$,
however the menus are chosen.  The full graph has no such restriction: its $h$ hubs
can use the roughly $h$ free columns wherever they happen to fall.

There is also an upper bound that holds for every hub instance: independent uniform menus lose at most $1/(k+1)$. Therefore, hub instances can not yield tighter asymptotic lower bounds.

\begin{proposition}[Hub instances are $\Theta(1/k)$]
\label{prop:hub-upper}
For every hub instance and every $k\ge1$, the menu rule that retains every dedicated
edge and gives each present hub an independent uniformly random set of $k$ columns
(all columns if $k\ge m$) satisfies
\[
    \mathbb E[\nu(H)]\ \ge\ \Bigl(1-\frac1{k+1}\Bigr)\mathbb E[\nu(\mathcal G)].
\]
\end{proposition}

\begin{proof}
If $k\ge m$, every edge is retained, so assume $k<m$.
Fix an arrival profile $T$. Let $f$ be the number of columns
whose dedicated arrivals are absent, and let $h'$ be the number
of present hubs. Put
\[
    s:=\min\{h',f\},
    \qquad
    r:=m-f+s=\nu(\mathcal G).
\]

First match every present dedicated arrival to its column.
Choose any $s$ present hubs in an order determined independently
of their menus. Expose their menus in that order, matching each
hub to an unused column in its menu whenever possible.

Before the $(j+1)$st hub is processed, at most $m-f+j$ columns
are occupied. Conditional on $T$ and all previously exposed menus,
the next menu is an independent uniform $k$-subset of the columns.
Thus the probability that this hub cannot be matched is at most
\[
    \frac{\binom{m-f+j}{k}}{\binom{m}{k}}
    \le
    \left(\frac{m-f+j}{m}\right)^k,
    \qquad 0\le j<s,
\]
where $\binom{a}{k}=0$ when $a<k$.

Let $L$ be the number of processed hubs left unmatched.
The resulting matching has size $r-L$, and hence
\[
\begin{aligned}
    \mathbb E[\nu(\mathcal G)-\nu(H)\mid T]
    &\le \mathbb E[L\mid T] \le
       \sum_{j=0}^{s-1}
       \left(\frac{m-f+j}{m}\right)^k \le
       \int_{m-f}^{r}\left(\frac{x}{m}\right)^k\,dx\\
    &\le
       \frac{r^{k+1}}{(k+1)m^k}
     \le \frac{r}{k+1}.
\end{aligned}
\]
The integral bound uses monotonicity of the integrand, and the
last inequality uses $r\le m$. Since $r=\nu(\mathcal G)$,
averaging over $T$ proves the claim.
\end{proof}

Combining \eqref{eq:intro-upper-rate} and
\eqref{eq:intro-lower-rate}, for all sufficiently large $k$,
\[
    \frac{k^k}{(k+1)^{k+1}}
    \le
    1-\alpha^\star(k)
    \le
    O\!\left(\frac{\log\log\log k}{\log\log k}\right).
\]
On hub instances, \Cref{prop:hub-upper} bounds the worst-case
loss between $k^k/(k+1)^{k+1}$ and $1/(k+1)$ for every $k\ge1$.
The ratio of these upper and lower bounds is
$(1+1/k)^k<e$, approaching $e$ as $k\to\infty$.

\section{Experimental details}
\label{app:experiments}

\paragraph{Scope of the experiments.}
The hub families model competition for resources whose availability is
uncertain.  Each offline vertex has a dedicated arrival adjacent only
to it, present with a specified probability.  Hubs are adjacent to all
these vertices and add value by using those whose dedicated arrivals
are absent.  NC is symmetric with one hub; NCH has heterogeneous
presence probabilities; MH$_4$ has four hubs; and MH$_k$ has $k$ hubs.
The private-resource variant TS$_4$ lets the same arrival have both a
contested private resource and many shared alternatives.

NC, NCH, MH$_4$ and MH$_k$ are hub instances, so
\Cref{prop:hub-upper} applies at every scale.
The construction of \Cref{thm:needy-lower} with $h=k$
corresponds to scale $c=1+1/k$ in MH$_k$,
which is not generally included in our tested grid.
For TS$_4$, consider the rule that retains the private vertex $c_i$
in the home state and samples $k-1$ uniform columns in either state.
Match private vertices first.  Their contribution is the same in the
full and retained graphs, and the remaining semi-hubs form an
independent random set of hubs with menus of size $k-1$.
\Cref{prop:hub-upper} therefore bounds this rule's loss by $1/k$.
The nested-neighborhood, random-presence, and random-heterogeneous
families are outside the scope of that proposition.
These analytical guarantees concern specified rules; the main plot
evaluates VarOpt with estimated marginals.  Its maxima over a finite
scale grid are not estimates of $\alpha^\star(k)$.

\paragraph{Families.}
All hub families have always-present hubs.
NC$(m)$: one hub, $\sigma_j=1-1/m$.
NCH$(m)$: one hub, $1-\sigma_j\propto$ log-uniform on $[1/4,4]$ with mean $1/m$.
MH$_h(m)$: $h$ hubs, $1-\sigma_j=h/m$.
TS$_h(m)$: the light layer of MH$_h(m)$, whose hubs are replaced by $h$ semi-hubs;
semi-hub $i$ has a private offline vertex $c_i$, contested by an intruder present with
probability $1/2$ and adjacent to $c_i$ alone, and two equiprobable states, ``home''
with neighborhood $\{c_i\}\cup\{v_1,\dots,v_m\}$ and ``away'' with neighborhood
$\{v_1,\dots,v_m\}$.
Scaled runs use $m=ck$ for NC and NCH, $m=4ck$ for MH$_4$ and TS$_4$, with
$c\in\{1.5,2,2.3,3,4\}$ and $k\in\{2,3,4,6,8,12,16,24,32,48,64\}$; MH$_k$ uses
$h=k$ and $m=ck^2$ with $c\in\{0.75,1,1.5,2.3\}$ for $k\le32$ and $c\in\{0.75,1\}$ for
$k\in\{48,64\}$.  In the main figure NC appears only through its exact loss.
Growing families outside the hub class (\Cref{fig:rate-appendix}): nested
neighborhoods (arrival $j$ present with probability $1-3/m$ and adjacent to
$v_0,\dots,v_j$; $m=ck$, $c\in\{1.5,2.3,4\}$); random presence ($30k$ arrivals present
with probability $0.8$, each with a fixed uniformly random neighborhood of size $4k$
among $18k$ offline vertices); random heterogeneous ($30k$ arrivals with three states
of Dirichlet$(1,1,1)$ probabilities floored at $0.05$, neighborhood sizes in
$\{1,2,4\}\cup\{4k\,2^{-i}:0\le i\le5\}$ weighted toward the largest, offline vertices
drawn with Zipf$(1)$ popularity among $21k$).
The sensitivity check and the selector comparison also use three fixed instances:
RH$(300)$, a random heterogeneous instance with $300$ arrivals and $210$ offline
vertices, three states per arrival with Dirichlet$(1,1,1)$ probabilities floored at
$0.05$, neighborhood sizes in $\{1,2,4,8,16,32\}$ with geometric weights favoring
the smaller sizes, and Zipf$(1)$ offline popularity; and SCAT$(300,d)$ for
$d\in\{8,32\}$, $300$ arrivals present with probability $0.8$, each with a fixed
uniformly random neighborhood of size $d$ among $180$ offline vertices.

\paragraph{Benchmark and marginals.}
The benchmark is Hopcroft--Karp on the realized graph with the arrivals in a uniformly
random order and the offline vertices relabeled by a uniformly random permutation;
this is a random maximum matching but not a uniform one.  Conditional marginals
$z_{itv}$ are empirical frequencies over $N_z$ unconditional arrival profiles: for
each state $(i,t)$ we divide the number of profiles in which $(i,v)$ is matched by the
number in which $T_i=t$.  We use $N_z=3000$, reduced to $2000$ for $k\ge48$ and for
MH$_k$ with $k\le32$, and to $1200$ for MH$_k$ with $k\ge48$.  Every state with a
nonempty neighborhood has probability at least $0.05$ in all families, so each
receives at least $60$ profiles in expectation.  The estimates can assign zero weight
to an edge with positive true marginal, so the implementation need not
satisfy (D1) relative to the true benchmark marginals.  The finite menu
banks described below introduce an additional approximation.

\paragraph{Menus.}
VarOpt$_k$ inclusion probabilities are $\min\{1,z_{itv}/\tau\}$ with
$\sum_v\min\{1,z_{itv}/\tau\}=k$, after adding $10^{-12}$ to every estimate so that
edges of zero estimated weight share the remaining budget uniformly; a neighborhood
with at most $k$ edges is retained entirely.  The probabilities are rounded by the
pivotal method with an independent uniformly random processing order for each
sample, which realizes them exactly, has pairwise nonpositive covariances, and, when
the probabilities are equal, yields a uniformly random $k$-subset by exchangeability.
For each state we draw $64$ menus once and, at evaluation, select one uniformly at
random at every occurrence of the state.  The exact inclusion probabilities
and covariance guarantees above describe the underlying sampler;
they need not hold exactly after conditioning on this finite menu bank.  The comparison rules are \emph{sim}, which
draws $k$ partners with replacement from the empirical partner pool of the state and
retains the distinct ones (possibly fewer than $k$); \emph{top-$k$}, the $k$ largest
estimated marginals with random tie-breaking; and \emph{uniform}, $k$ uniformly random
neighbors.

\paragraph{Estimation of the loss.}
Marginal estimation and evaluation use independent profiles.  For each of
$N_{\rm ev}$ evaluation profiles ($1500$; $1200$ for $k\ge48$ and for MH$_k$ with
$k\le32$; $800$ for MH$_k$ with $k\ge48$) we compute $\nu(\mathcal G)$ and $\nu(H)$ for
every rule on the same profile; the loss is the mean difference divided by the mean of
$\nu(\mathcal G)$, and the reported standard error is that of the mean difference,
divided by the mean of $\nu(\mathcal G)$.
This is an approximate evaluation standard error that treats the
denominator as fixed and conditions on the estimated marginals and
menu banks.  It omits their estimation uncertainty and is not adjusted
for selecting the largest estimated loss over the scale grid.
Thus the plotted bars are not simultaneous confidence bands for the
maxima.
On NC, where every menu with $k$ entries has loss $(a^k-a^m)/(m(1-\rho)+1-a^m)$ with
$a=1-\rho$, $\rho=1/m$, the estimator reproduces the exact value within two standard
errors at all $49$ tested $(k,c)$ pairs.  Re-estimating the marginals with three
independent seeds and with $N_z=12000$ on RH$(300)$, TS$_4(64)$, MH$_4(48)$ and
NCH$(18)$ changes the VarOpt loss by at most $6\%$ (relative), within about two
evaluation standard errors; the deterministic top-$k$ rule on NCH varies by up to
$25\%$ across seeds, since its menus change discontinuously with the estimates.
Losses below two standard errors are omitted from the log-scale plot.

\paragraph{Number of hubs.}
At $k=8$, $k$ times the loss of VarOpt menus on MH$_h$ at the best tested scale
($c\in\{0.5,0.75,1,1.25,1.5,2,2.3,3\}$ for $h\le32$, $\{0.75,1,1.25\}$ for
$h\in\{64,128\}$) is $0.12$, $0.20$, $0.27$, $0.32$, $0.36$, $0.39$, $0.41$, $0.42$ for
$h=1,2,4,\dots,128$.  The tested procedure samples each hub's menu independently and
does not preassign disjoint blocks of columns, which the model would allow.

\begin{figure}[t]
\centering
\includegraphics[width=0.55\textwidth]{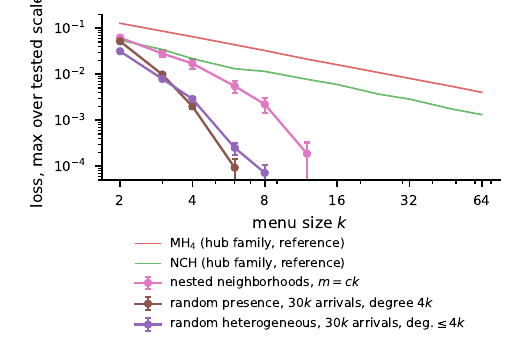}
\caption{Loss of VarOpt$_k$ menus on growing families outside the hub class, maximum
over the tested scales, with two hub families of \Cref{fig:rate} for reference; error
bars show twice the estimated evaluation standard error, and points below that level are omitted.}
\label{fig:rate-appendix}
\end{figure}

\paragraph{Rules and benchmark selector.}
\Cref{tab:rules} compares the four rules at $k\in\{4,8\}$ (maximum over tested scales for the scaled and growing families).  Uniform menus have substantially larger measured loss on the heterogeneous
families and can omit the private resource in TS$_4$.
On symmetric NC, every menu with $k$ entries has the same exact loss;
the differences between full-menu rules in the table are sampling noise.
Top-$k$ has the smallest measured loss on NCH and the random
heterogeneous family, but the largest on MH$_4$, consistent with
collisions between hubs' preferred choices.
The simulation menu can retain fewer than $k$ distinct edges and has
larger measured loss than VarOpt in these comparisons.
Replacing the Hopcroft--Karp selector by a maximum matching that, among maximum
matchings, prefers low-degree arrivals (dense assignment with edge weights
$1+1/(2n\deg i)$) roughly halves the VarOpt loss on the fixed random heterogeneous
instance RH$(300)$ at every $k$ at which it is measurable ($31.5$ to $19.9$,
$3.6$ to $2.0$ and $0.69$ to $0.30$, all $\times10^{-3}$, at $k=2,4,6$), increases
it by $10$--$20\%$ on the random presence graphs, and changes it by at most $7\%$ on
the hub instances NC$(64)$ and TS$_4(64)$.  These comparisons show that the maximum-matching selector can
materially affect the measured loss; they do not determine an
asymptotic decay rate.

\begin{table}[t]
\centering\scriptsize
\caption{Loss $\times10^{3}$ (estimated evaluation standard error) by rule}
\label{tab:rules}
\begin{tabular}{llrrrr}
\toprule
family & $k$ & VarOpt & sim & top-$k$ & uniform \\
\midrule
NC & 4 & 33.6 (1.4) & 46.5 (2.0) & 32.5 (1.5) & 32.2 (1.3) \\
NC & 8 & 15.0 (0.7) & 24.4 (1.0) & 15.2 (0.7) & 15.8 (0.7) \\
NCH & 4 & 21.7 (1.0) & 45.1 (2.0) & 18.0 (0.7) & 35.5 (1.4) \\
NCH & 8 & 11.5 (0.4) & 20.4 (0.9) & 10.3 (0.4) & 16.4 (0.7) \\
MH$_4$ & 4 & 65.6 (1.1) & 68.7 (1.1) & 114.3 (1.1) & 63.4 (1.0) \\
MH$_4$ & 8 & 32.7 (0.5) & 33.9 (0.5) & 51.3 (0.6) & 30.6 (0.5) \\
TS$_4$ & 4 & 58.6 (0.9) & 64.7 (1.0) & 65.8 (1.0) & 67.7 (1.0) \\
TS$_4$ & 8 & 25.9 (0.5) & 29.6 (0.5) & 33.4 (0.5) & 35.1 (0.5) \\
nested & 4 & 17.1 (2.1) & 24.4 (1.8) & 43.9 (3.3) & 92.9 (2.2) \\
nested & 8 & 2.2 (0.4) & 7.7 (0.8) & 3.3 (0.5) & 65.6 (1.2) \\
random presence & 4 & 2.0 (0.1) & 19.3 (0.4) & 0.9 (0.1) & 3.8 (0.2) \\
random presence & 8 & 0.0 (0.0) & 0.5 (0.1) & 0.3 (0.0) & 0.0 (0.0) \\
random heterogeneous & 4 & 2.9 (0.2) & 24.1 (0.4) & 1.1 (0.1) & 80.9 (0.7) \\
random heterogeneous & 8 & 0.1 (0.0) & 1.2 (0.1) & 0.0 (0.0) & 12.2 (0.2) \\
\bottomrule
\end{tabular}
\end{table}

\end{document}